\documentclass[final,onefignum,onetabnum]{siamart190516}

\usepackage{lipsum}
\usepackage{amsfonts}
\usepackage{amssymb}
\usepackage{graphicx}
\usepackage{epstopdf}
\usepackage{algorithmic}
\usepackage{todonotes}
\usepackage{stmaryrd}
\usepackage{subfig}

\ifpdf
  \DeclareGraphicsExtensions{.eps,.pdf,.png,.jpg}
\else
  \DeclareGraphicsExtensions{.eps}
\fi

\newsiamremark{remark}{Remark}
\newsiamremark{hypothesis}{Hypothesis}
\newsiamthm{fact}{Fact}
\newsiamthm{defn}{Definition}
\crefname{hypothesis}{Hypothesis}{Hypotheses}
\newsiamthm{claim}{Claim}
\newsiamthm{lem}{Lemma}
\newsiamthm{thm}{Theorem}
\newsiamthm{prop}{Proposition}
\newsiamthm{cor}{Corollary}
\headers{Optimal Mixing in Transport Networks}{C. Mentus and M. Roper}

\title{Optimal Mixing in Transport Networks: Numerical Optimization and Analysis\thanks{Submitted to the editors on \today.
\funding{This work was funded by the National Science Foundation under grant no.~DMS--1351860.}}}

\author{Cassidy Mentus\thanks{Dept. of Mathematics, University of California, Los Angeles, CA 90095
  (\email{cassidy.mentus@gmail.com}, \email{mroper@math.ucla.edu}, \url{www.marcusroper.org}).}
\and Marcus Roper \footnotemark[2]}

\usepackage{amsopn}

\ifpdf
\hypersetup{
  pdftitle={Optimal Mixing in Transport Networks},
  pdfauthor={C. Mentus and M. Roper}
}
\fi

\begin{document}

\maketitle

\begin{abstract}
Many foraging microorganisms rely upon cellular transport networks to deliver nutrients, fluid and organelles between different parts of the organism. Networked organisms ranging from filamentous fungi to slime molds demonstrate a remarkable ability to mix or disperse molecules and organelles in their transport media. Here we introduce mathematical tools to analyze the structure of energy efficient transport networks that maximize mixing and sending signals originating from and arriving at each node.  We define two types of entropy on flows to quantify mixing and develop numerical algorithms to optimize the combination of entropy and energy on networks, given constraints on the amount of available material. We present an in-depth exploration of optimal single source-sink networks on finite triangular grids, a fundamental setting for optimal transport networks in the plane.  Using numerical simulations and rigorous proofs, we show that, if the constraint on conductances is strict, the optimal networks are paths of every possible length. If the constraint is relaxed, our algorithm produces loopy networks that fan out at the source and pour back into a single path that flows to the sink. Taken together, our results expand the class of optimal transportation networks that can be compared with real biological data, and highlight how real network morphologies may be shaped by tradeoffs between transport efficiency and the need to mix the transported matter. \end{abstract}



\section{Introduction}
\label{sec:intro}
Work by Murray in the 1920s \cite{murray1926physiological} first probed the idea that vessels in biological transportation networks may optimize knowable target functions. Murray hypothesized that blood vessels may have optimal radii are set by tradeoffs between the need to minimize friction within the vessel (which favors large vessels), and the energetic cost of maintaining the vessel (which penalizes large vessels). The scalings and geometric relationships that he derived from this trade-off have found some experimental support for the blood networks of animals \cite{zamir_arter_geom} and water transport networks of plants \cite{PlantTransportMcCullohAdler}. More recent theoretical work has extended the analysis of single vessels or branch points to whole networks of vessels in which the sources and sinks are prescribed but the network is given many choices for how to connect these points \cite{bohn2007structure,durand2007structure}, added damage or fluctuations in source and sink strengths \cite{katifori2010damage,corson2010fluctuations}, or developed models for how feedbacks between flows and network growth allow such optimal networks to be grown \cite{hu2013adaptation,ronellenfitsch2016global}.

Hundreds of thousands of species of microorganisms, including slime molds, water molds and fungi rely on internal transportation networks. These networks have similar functions -- they continuously grow as the organism claims territory or searches for hosts or resources. Within the network nutrients, fluid and cellular matter (including nuclei and other organelles) are transported from sites of production or uptake to sites of utilization. Minimization of friction, in conjunction with robustness to damage, appears to underlie features of some of the foraging networks made for example by wood rotting basidiomycete fungi \cite{BoddyTransportNetwork} and slime molds \cite{NakagakiAdaptive}. However, organisms build networks with a tremendous diversity of morphologies that can not be explained by friction minimization alone. Do these morphologies emerge from other physical principles besides minimizing friction, from constraints on the pathways used to grow the network, or from neutral differences in network morphology that do not affect the organism's fitness? We start from the position that to understand the extent of the role that optimization plays in determining the structure of networks, we must first understand what the optimal network is for a given target function. This approach previously guided us to develop gradient-descent methods for optimizing networks for arbitrary differentiable functions \cite{chang2019microvscular}.

In this work we focus on a quantity with many points of non-differentiability: the amount of mixing occurring within the network. This quantity, which is given two different quantifications below, is non-differentiable in the conductances of the network at any point where the flow in an edge goes to 0. Since the optimization of the network requires searching over possible topologies for flow; i.e. reversing the directions of flow on edges, we develop here a new numerical optimization method that is adapted to deal with this pervasive non-differentiability.

Why are might real networks seek to maximize mixing? Three kinds of mixing seem to be relevant to network-forming microorganisms:

1. In fungal networks cellular growth occurs at the periphery of the network through the continuous extension of hyphae at their tips, and in fast growing fungi, such as the model organism {\it Neurospora crassa} growth requires the continuous supply of nuclei and other organelles to the edge of the mycelium \cite{LewHyphaGrow}. Within {\it N. crassa} nuclei often take tortuous and multidirectional paths toward the tips, and the network is known to be organized so that pairs of nuclei that start close together within the mycelium are unlikely to be delivered to the same site of growth at the periphery, potentially to stop deleterious mutations accumulating in one region of the fungus \cite{RoperNuclear}.

2. Recent experiments in the dung fungus {\it Coprinopsis cinerea} show large swathes of the network responding to the external threat of predatory nematode worms. When nematode grazing is detected in one part of the fungal network, a suite of defense chemicals is expressed, not just at the site of grazing, but spreading in multiple directions through the network \cite{Nematoxic}. Spreading out nematoxin production may prepare other parts of the network for further attacks or enable the cost of labor to be spread through the network \cite{RoperCurrBiology}.

3. Plasmodial slime molds, such as {\it Physarum polycephalum} live in heterogeneous environments containing patches of nutrients \cite{AlimPhysarumPrimer}. The network remodels globally when it discovers a new nutrient source, and it is thought that individual tubes in the network respond to a cue carried within the flow carried within the network \cite{alim2017mechanism}. A global response to this cue requires that it be dispersed through the entire network.

We model the signals within the network as being passively transported by the flows. In Section \ref{sec:entropy} we define an entropy of mixing of the transported signals in a flow network. In Section \ref{sec:numerics} we describe a numerical method for choosing the conductances within the network, and in Section \ref{sec:results} we show simulation results. A highlight result is that for small values of the parameter, $\gamma$, which represents the penalty of dividing one edge into two, the optimal networks become a set of paths linking source and sink. We prove why paths are favored, and analytically expose the set of possible path optima in Section \ref{sec:theory}.

\section{Mathematical model and mixing entropies}
\label{sec:entropy}

\subsection{Movement of signals through a flow network}
\label{sec:movement}
Our mathematical model for the biological transportation network consists of a network (graph) with nodes (vertices)  $\mathcal{N}$, enumerated $1,2,\ldots,N$, and edges $\mathcal{E}$. The nodes are arranged on a regular triangular lattice, so that each vertex in the interior is linked by equal length edges to 6 neighbors (we write $n(i)$ for the set of neighbors of $i$). The conductance of the edge $(i,j)$ is denoted by $\kappa_{ij}$. Fluid (protoplasm) is continually pushed through the network by pressure differences between the nodes. In our model the ultimate origin of these pressure differences are flows into and out of the network via diametrically opposite nodes. Signals are carried along by this bulk flow of fluid.
\begin{figure}
\includegraphics[width=1\textwidth]{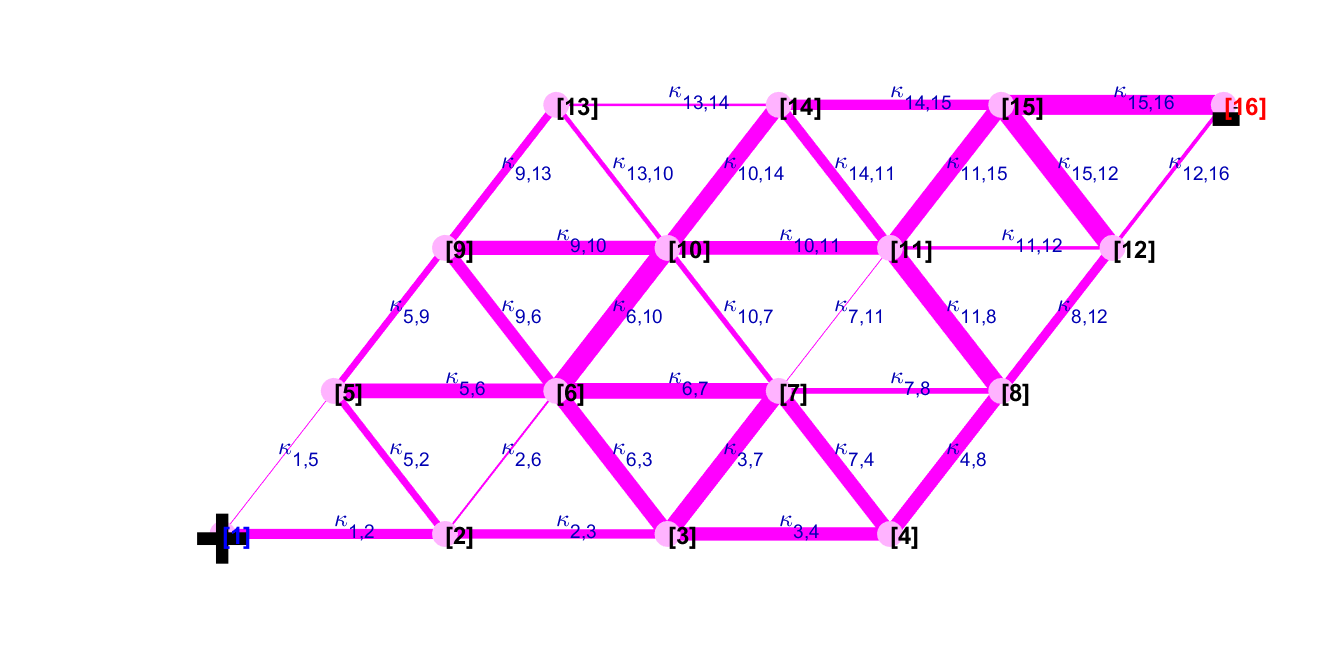}
\caption{Conductance network $\kappa_{ij}$ with each node and edge conductance labeled. The source is at node $1$, the bottom-left corner, and the sink is at node $25$, the top right corner.}
\end{figure}
\begin{definition}
The rate of fluid entering or exiting the network at $i\in \mathcal{N}$ is the \textbf{boundary flow at node $i$}, and is denoted $Q_{i}$.
\end{definition}
Boundary flow $Q_{i}>0$ corresponds to fluid entering the network through node $i$ (i.e. the node is a source), and $Q_{i}<0$ corresponds to fluid exiting the network at node $i$ (i.e. the node is a sink).  The total volume of fluid contained in the network is constant, so total inflows and outflows must be balanced: $\sum_{i}Q_{i}=0$. The boundary flows in turn engender flows, $q_{ij}$, on the edges. Flows must also be balanced on each node in the network, a fact that is known as Kirchhoff's first law of circuits:
\begin{definition}
A flow is called \textbf{compatible} with regards to the boundary
flows $Q_{i}$ if $\sum_{j\in n(i)}q_{ij}=Q_{i}$ for all $i\in\mathcal{N}$.
\end{definition}
For any set of boundary flows, there are typically multiple compatible flows on the network. The flow we are interested in, called the \textbf{physical flow}, is the unique compatible flow that minimizes the dissipation:
\begin{definition}
For a conductance network $\kappa_{ij}$, the \textbf{dissipation
$\mathcal{D}$ } from flows $q_{ij}$ is the rate at which work must be done to maintain the fluid flows on all edges of the network: $\mathcal{D}(q_{ij})=\sum_{ij}\frac{q_{ij}^{2}}{\kappa_{ij}}$
. \label{def:dissipation}
\end{definition}
The flow that minimizes the dissipation can be derived from Kirchhoff's first and second laws \cite{chang2018minimal}, which introduce a pressure variable that is defined on each node in the network:
\begin{proposition} {\bf Kirchoff's second law for circuits} \label{thm:formulaforflows}Let $\kappa_{ij}$ be a connected conductance
network with nodes $\mathcal{N}$ and edges $\mathcal{E}$. Let $Q$ be boundary flows such that $\sum_{i}Q_{i}=0$. Let $q_{ij}$ be the physical flows
of this network. Then there exists $p_{i}\in\mathbb{R}$, called the \textbf{\emph{pressure}} at
node $i$, such that $q_{ij}=\kappa_{ij}(p_i-p_j)$. \label{prop:Kirchhoff2}
\end{proposition}
We can compute the pressures by defining a vector of pressures $\mathbf{p}=\left\{p_i\right\}_{i\in\mathcal{N}}$, a vector of boundary flows $\mathbf{Q}=\left\{Q_i\right\}_{i\in\mathcal{N}}$ and the network Laplacian $\Delta_\kappa$, a $|\mathcal{N}|\times |\mathcal{N}|$ matrix with entries \begin{equation}
\Delta_{\kappa,ij}=\left\{
\begin{array}{cc} 
-\kappa_{ij} & \hbox{if~} i\not=j \\ 
\sum_{j}\kappa_{ij}& \hbox{if~} i=j\end{array}\right.\label{eqn:Laplacian}
\end{equation}
Then conservation of mass at each node (Proposition \ref{prop:Kirchhoff2}) is equivalent to solving \begin{equation}
\Delta_{\kappa}\mathbf{p} = \mathbf{Q}.\label{eqn:Poiss_eqn}
\end{equation}
So long as every connected component of a physical network has one node with a defined pressure, the pressures are uniquely solvable, otherwise they are solvable up to a single additive constant per connected component\cite{chang2018minimal}. When the conductance network is connected and the pressure at node $i$ is known: $p_i = P$, we add $P$ to $Q_i$ and construct the invertible matrix $\tilde{\Delta}_{\kappa}$ by adding $1$ to $\Delta_{\kappa,ii}$.

We now consider the mixing produced by the flows within the network. Our flow network model contains all of the scenarios for mixing described in Section \ref{sec:intro}. The signals passing through the network could represent genetically diverse nuclei (scenario 1), or chemical cues (scenarios 2 and 3). Our model does not need to represent the entire network, it could represent the portion of network that supplies a single hyphal tip. This supply network would be linked to supply networks for other tips, and acquires signals, randomly at each node from these other networks\footnote{Signals can be transferred between supply networks without flow between them, since motor protein trafficking (of nuclei) or diffusion (of chemical cues) provide alternate transport mechanisms.}. Signals are made up of blobs: either molecules or organelles. We compile a list of the nodes visited by each signal blob: call the $t$-th node visited by a signal, $x_t$. Then $x_t$ is a random walk, with transition probability:
\begin{equation}
T_{ij}  \equiv  P(x_{t+1}=j|x_{t}=i)=\frac{q_{ij}}{\sum_{k\in n(i):\,q_{ik}>0}q_{ik}-Q_{i}\boldsymbol{1}_{Q_{i}<0}}.
\end{equation}
that is, the flow of signal from $i$ to $j$ is simply proportional to the total flow along the edge $(i,j)$. Effectively we assume that signal is uniformly dispersed in the flowing protoplasm, ignoring any physical effects such as diffusion \cite{MarbachAlimPruning4Taylor} that move signals independently of flows. When the signal reaches a sink node it may exit the modeled network (with the exit probability proportional to $-Q_i$, so $\sum_{j}T_{ij}\leq 1$).

A signal introduced at node $i$ travels along the network following the flow. At each node with multiple possible outward flows, the signal chooses one outflow probabilistically. Signals therefore perform a type of random walk down the pressure gradient. There are two senses in which signals may be considered to mix on the network: 1. Given the node $i$ at which it originates we are interested in the number of nodes that the signal visits before exiting the network. 2. Alternately, given a node $j$, we are interested in the number of different sites of origin that signals passing through $j$ may have. To quantify either form of diversity, we must focus on the probability that a signal originating at node $i$ ever visits a node $j$ defined by:
\begin{equation}
P_{ij}  =  P(x_{t}=j~\textrm{ for some }t\geq0~|~x_{0}=i).
\end{equation}
The entries of $P_{ij}$ from a $N\times N$ matrix. To calculate this matrix from the transition probabilities $T$, note that the probability of getting from $i$ to $j$ by following exactly $n$ edges is $(T^n)_{ij}$. Hence,
$P = \sum_{n=0}^{N-1} T^n$ (note that $T^N=0$, because a signal can visit at most $N$ nodes before exiting the network and signals can not visit the same node twice). Alternatively by summing the geometric progression:
\begin{equation}
P = (I - T)^{-1}~,
\end{equation}
where $I$ is the $N\times N$ identity matrix.

\subsection{Defining mixing entropies}

We define two types of information entropy on the flows $q_{ij}$. The first is a measure of the accumulation of signals at every node in the network and the second represents the dispersal of signals throughout the network. We call the two entropies, respectively, {\bf total receiver entropy} (or {\bf total mixing entropy}) and {\bf total sender entropy}. Let $f_{i}$ be the total flow through node $i$, i.e. $f_{i}=\sum_{j\in n(i):q_{ij}>0}q_{ij}+Q_{i}\boldsymbol{1}_{Q_{i}>0}$. The rate at which fluid flows from $i$ to $j$ is then $\tilde{q}_{ij}=P_{ij}f_{i}$. We refer to this as the flow from $i$ to $j$. We assume that the rate at which a signal is produced at a node is proportional to the total flow through that node. This assumption certainly makes sense if our signal consists of new nuclei that are generated by divisions within the protoplasm, since the flow through a node will be proportional to the rate at which nuclei pass through it. For other signal production scenarios (such as when the signal is produced in response to predation), we can arrive at this assumption if we assume that product of the new signal is rate-limited by a component that is contained within the protoplasm, so signal production rate is proportional to rate of protoplasm cleared through the node in unit time. Under this assumption the relative proportions of signals received at node $j$ from upstream nodes $i$ are the same as the relative proportions of $\tilde{q}_{ij}$. 

In our model each site in the network can send signals to other sites in the network, and any point in the network may potentially receive signals from any other point. We cannot tell ahead of time which nodes will provide the useful signals, so we consider all nodes as possible sources of signals. We also make no assumption about sites where diversity needs to be maximized (this is in contrast to \cite{RoperNuclear}, in which genetic diversity was considered only at hyphal tips), so we consider all of the possible nodes that signals can reach within the network when computing the mixing entropy.

To compute the entropy of the distribution of signals arriving at $i$ we define the probability distribution on up-stream nodes of $i$: 
\begin{equation}
\mathcal{P}_{i}(j)  =  \frac{\tilde{q}_{ji}}{N_i}~~\hbox{where}~~N_i\equiv \sum_{j:\tilde{q}_{ji}>0}\tilde{q}_{ji}~,
\end{equation}
effectively forming a new matrix from $\tilde{q}_{ij}$ in which all columns are normalized to sum to 1. We may define the \textbf{local receiver entropy at node $i$ }as the Shannon information entropy of $\mathcal{P}_i$: $H(\mathcal{P}_{i})=-\sum_{j}\mathcal{P}_{i}(j)\log\left(\mathcal{P}_{i}(j)\right)$. We consider the total flow through $i$ as a measure of the ``importance'' of the node \cite{izsak2007parameter}. In our model, the diversity of signals is more important at high traffic nodes than at low traffic nodes. This principle is useful mathematically, since it ensures that rearrangements of very low conductance edges don't greatly affect the overall mixing associated with a network. At the same time, the weighting is intended to reflect the relative biological importance of nodes within the network -- a node with
high flow supplies a greater volume of cytoplasm to the rest of the network, so it is more important that all of the signals (whether
cues or nucleotypes) are present at the node. Hence the total receiver entropy is:
\begin{equation}
H =  \sum_{i}f_{i}H(\mathcal{P}_i).
\end{equation}
Similar to \cite{tanyimboh1993calculating} $H$ represents the conditional entropy associated with choosing a {\it receiving node} at random with probability proportional to $f_i$ and then conditioned on our choice of node $i$ we chose a sending node at random via the distribution $\mathcal{P}_{i}$. 

\subsection{Set restrictions of the entropy}

Our proofs in Section \ref{sec:theory} often require that we partition $\mathcal{N}$ into subsets of nodes. It is convenient to be able to evaluate the contributions of each subset to the total network entropy. We define restricted entropies for subsets $\mathcal{F}\subset\mathcal{N}$ as follows: For all $i\in\mathcal{F}$ define $\mathcal{P}_{\mathcal{F}i}(j)=\frac{\tilde{q}_{ji}}{\sum_{k\in\mathcal{F}}\tilde{q}_{ki}}$
if $j\in\mathcal{F}$ and $\mathcal{P}_{\mathcal{F}i}(j)=0$ otherwise.
\begin{defn}
The local negative mixing entropy restricted to $\mathcal{F}$ is
defined to be 
\begin{equation}
H(\mathcal{P}_{\mathcal{F}i}) =- \sum_{j\in\mathcal{F}:\tilde{q}_{ji}>0}\mathcal{P}_{\mathcal{F}i}(j)\log\left(\mathcal{P}_{\mathcal{F}i}(j)\right)
\end{equation}
 and the total mixing entropy restricted to $\mathcal{F}$
is 
\begin{equation}
    H_{\mathcal{F}}  =  \sum_{i\in\mathcal{F}}f_{i}H(\mathcal{P}_{\mathcal{F}i}) ~.
\end{equation}
\end{defn}
%

\subsection{Sending entropy on flows}

It may also be important for the network to spread out signals to as many downstream nodes as possible. We define an entropy for the places that can be reached by a new signal originating at a node within the network. Specifically, instead of taking the mass distribution of incoming flows and normalizing them to a probability distribution, we use the out-going flows. That is
we define the probability distribution of nodes that can be reached by a signal introduced at node $i$:
\begin{equation}
\mathcal{P}'_{i}(j)  =  \frac{\tilde{q}_{ij}}{\sum_{j:\tilde{q}_{ij}>0}\tilde{q}_{ij}}.
\end{equation}
This is equivalent to normalizing the matrix $\boldsymbol{\tilde{q}}$ so that all rows sum to 1. We define the \textbf{local sending entropy at node $i$} to be
the Shannon information entropy \cite{shannon1948mathematical} of the distribution $\mathcal{P}'_i$: 
\begin{equation}
H(\mathcal{P}'_{i})  =  -\sum_{j}\mathcal{P}'_{i}(j)\log(\mathcal{P}'_{i}(j))~.
\end{equation}
and we define the total sending entropy of the entire network to be the weighted sum of the node entropies:
\begin{equation}
H'  =  \sum_{i}f_{i}H(\mathcal{P}'_{i})~.
\end{equation}

\subsection{Equivalence of receiving and sending entropies}

Although the entropies $H$ and $H'$ offer alternate representations of the mixing that occurs within the network, they are linked by an equivalence principle:
\begin{theorem}
Let $q_{ij}$ be a flow network compatible with boundary flows $Q_{i}$. Let $q'_{ij}$ and $Q'_{i}$ be the flow network and boundary flows
obtained from $q_{ij}$ and $Q_{i}$ by reversing the flows, i.e. $q'_{ij}=-q_{ij}$ and $Q'_{i}=-Q_{i}$. Then $H'(q_{ij})=H(q'_{ij})$
\label{thm:NSE=NME}.
\end{theorem}
\begin{proof}
Notice that reversing the flows doesn't affect the flow strengths of nodes within the network because 
\begin{eqnarray} 
f_i=\sum_{j:q_{ij}>0}q_{ij}+|Q_{i}|\boldsymbol{1}_{Q_{i}<0}&=&\sum_{j:q_{ij}<0}|q_{ij}|+|Q_{i}|\boldsymbol{1}_{Q_{i}>0} \\
&=&\sum_{j:q'_{ij}>0}|q'_{ij}|+|Q'_{i}|\boldsymbol{1}_{Q'_{i}<0}~.
\end{eqnarray}
The equivalence principle boils down to proving the statement $\tilde{q}_{ij}=\tilde{q}'_{ji}$ where $\tilde{q}'_{ij}$ is the flow from node $i$ to node $j$ in the flow-reversed network. We derive this equality by comparing the probability of a signal path $x_{t}$: $t=0,1,\ldots,T$ under the flow $q_{ij}$ with the probability of the reversed path $x'_t \equiv x_{T-t}$: $t=0,1,\ldots ,T$ under the reversed flow $q'_{ij}$:
%
since $|q_{ij}|=|q'_{ij}|$ it follows that $T_{ij}f_{i}=T'_{ji}f_{j}$
and so $T'_{ji}=\frac{f_{i}}{f_{j}}T_{ij}$. We multiply the probability
of the path $x_t$ by the strength of the starting node, $f_{x_{0}}$, to obtain $f_{x_{0}}\prod_{t=0}^{T-1}T_{x_{t}x_{t+1}}$, and rewrite via a telescoping product:
\begin{equation}
f_{x_{0}}\prod_{t=0}^{T-1}T_{x_{t}x_{t+1}}  =  f_{x_{T}}\prod_{t=0}^{T}\frac{f_{x_{t}}}{f_{x_{t+1}}}T_{x_{t}x_{t+1}} =  f_{x_{T}}\prod_{t=0}^{T}T'_{x_{t+1}x_{t}} =  f_{x_{T}}\prod_{t=0}^{T}T'_{x'_{t}x'_{t+1}}.
\end{equation}
For any nodes $i$ and $j$ in the network, we can sum over the
probability of all possible paths $i$ to $j$ in the regular network and $j$ to $i$ in the flow-reversed network to obtain: $\tilde{q}_{ij}=f_{i}P_{ij}=f_{j}P'_{ji}=\tilde{q}'_{ji}$. Hence the distributions
$\mathcal{P}'_{i}(j)$ for the flow network $q_{ij}$ are equal to
the distributions $\mathcal{P}_{i}(j)$ for the network $q'_{ij}$, so $H_i'(q_{ij}) = H_i(q'_{ij})$ leading to the required result.
\end{proof}
The physical flow on the network (see Proposition \ref{prop:Kirchhoff2}) can be reversed by reversing the
sources and sinks in the network; that is, replacing a source with
inflow $Q_{i}$ by a sink with outflow $Q_{i}$, and conversely. In the cases that we will analyze in this paper, the sources and sinks are matched in number and strength (e.g. a single source and single sink at opposite corners of a square grid network); so a network that optimizes receiving entropy can be transformed into a network that optimizes sending entropy simply by rotating the source into the sink and conversely. For this reason, we do not have to develop separate results for the two entropies. We focus on analyzing the receiving entropy, which we refer to simply as mixing entropy henceforth.

\subsection{Mathematical formulation of the optimization problem}

Building and using flow networks requires energy investment; an organism's optimal network will reflect tradeoffs between mixing effectiveness and the cost of the network. The cost has two components: each edge in the network must be built and maintained, and the fluid transported within the network dissipates energy
due to friction. The two cost components play slightly different roles in our optimization, we incorporate the first cost via a holonomic constraint, and the second via a penalty.

Murray \cite{murray1926physiological} posited that the cost of a maintaining a vessel is either proportional to its volume or surface area. Since all of the vessels in our networks have the same length, and the Hagen-Poiseuille law states that conductance is proportional to the fourth power of the radius, these scenarios correspond respectively to the cost of an edge being proportional to $\kappa_{ij}^{1/2}$ or to $\kappa_{ij}^{1/4}$. We constrain the cost the total material available to the network, requiring $\sum \kappa_{ij}^\gamma=C$ where $0<\gamma<1$ is a constant, whose effect on network morphology will be explored \cite{bohn2007structure,chang2019microvscular,FrickerSlimeMold}.

To incorporate the cost of dissipation in our optimization, we formulate it as a minimization problem:
\begin{equation}
\hbox{Find:~} \arg\min \{\Theta(\kappa_{ij}) \equiv -H(\kappa_{ij})+cD(\kappa_{ij})~:~~\kappa_{ij}\geq 0 ~\forall (i,j)\in\mathcal{E}~,~\sum \kappa_{ij}^\gamma = C\}. \label{eq:optimfn}
\end{equation}
We refer to $\Theta$ as the {\bf mixing-dissipation cost} (abbreviated: {\bf CMD}). Since the set of allowed conductances is compact, we know that the minimizer exists. The constant $c$ represents the relative priority to the network of minimizing dissipation over maximizing mixing. Along with $\gamma$ it is one of the main parameters that we explore in this work. We stretch our notation to refer to the minimum value of $\Theta$ for a given value of $c$ as $\Theta(c)$. 
\begin{lemma}
The minimal mixing-dissipation cost is a concave function of $c$.
That is, for $c_{1},c_{2}\geq0$: $\Theta(tc_{1}+(1-t)c_{2})\geq t\Theta(c_{1})+(1-t)\Theta(c_{2})$
for all $0\leq t\leq1$.
\end{lemma}
\begin{proof}
Set $c_3=tc_1+(1-t)c_2$, and let $\kappa_i$ be a minimizer of $-H+c_iD$; $i=1,2,3$. Then
\begin{eqnarray}
-tH(\kappa_{1})+tc_{1}D(\kappa_{1}) -(1-t)H(\kappa_{2})+(1-t)c_{2}D(\kappa_{2}) & \leq & -tH(\kappa_{3})+tc_{1}D(\kappa_{3})\nonumber \\ 
& & -(1-t)H(\kappa_{3}) \nonumber\\
& &+(1-t)c_{2}D(\kappa_{3})\nonumber \\
& = & -H(\kappa_{3})\nonumber \\
& &+\left(tc_{1}+(1-t)c_{2}\right)D(\kappa_{3}).
\end{eqnarray}
\end{proof}

To finish formulating the optimization problem we restrict the set of network topologies that we are searching over: Let $G$ be an unweighted undirected network with nodes $\mathcal{N}$ and edges $\mathcal{E}$. Choosing $\mathcal{E}$ allows us to constrain e.g. the maximum degree of the nodes in our optimal network. Our optimal network is restricted to be a subnetwork of $G$: we refer to $G$ as the {\bf ambient network}. For the purposes of this study we will assume that the network is planar (this assumption is almost certainly true for slime mold networks, but is less valid in fungal networks, where hyphae often crossover without connecting to each other). In this paper we restrict to regular triangular networks, in which all of the edges in the ambient network have the same length. We do not think that our results are sensitive to the choice of (regular) ambient network: we have for example, reproduced all of the results discussed in this paper with square grid ambient networks \cite{PhDThesis}.

\subsection{Invariance of optima to changing the material investment in the network} \label{sec:Cinvariance}

Solutions of our optimization problem depend upon the value of $C=\sum \kappa_{ij}^\gamma$. However, as $c$ is increased from $0$ to $\infty$, the same sequences of optimal networks are found, independent of $C$. For suppose that $\kappa_{ij}$ is a conductance network that solves Eq. \ref{eq:optimfn} with $\sum \kappa_{ij}^\gamma = C$. Then rescaling $ \kappa_{ij}'=\left(\frac{C'}{C}\right)^{1/\gamma}\kappa_{ij}$ produces a new network with $\sum \kappa_{ij}'^\gamma=C'$. The flows are unaltered in this network, so $H(\kappa_{ij}') = H(\kappa_{ij})$. However, dissipation is changed: $D(\kappa_{ij}')=\left(\frac{C}{C'}\right)^{1/\gamma}D$. So the new network minimizes $\theta$ in Eq. \ref{eq:optimfn} for the new dissipation weighting $c'=\left(\frac{C'}{C}\right)^{1/\gamma}c$. Sweeping through all values $c\geq 0$, with $\sum \kappa_{ij}^\gamma=C$, we generate in one-to-one correspondence all optimal networks for $c'\geq 0$ with $\sum \kappa_{ij}^\gamma = C'$. The choice of the value for $C$ is therefore arbitrary.

\section{Numerical optimization}
\label{sec:numerics}

Chang and Roper \cite{chang2019microvscular} optimized networks for general differentiable functions using gradient descent. However, the mixing entropy that we seek to optimize here is non-differentiable wherever the flow through an edge is equal to 0. It is necessary that the optimization algorithm be able to navigate through such points, because as conductances are updated to maximize mixing entropy it is often necessary to reverse the direction of flow on one or more edges. In Fig \ref{fig:mixinglandscape}, we show a contour map of varying two edge conductances within a network (the original network is shown at top right, and networks with reversed flow bottom and left). The landscape is tiled into watersheds, each watershed represents the set of entropies that can be attained by varying the conductances without reversing the direction of flow on any edge. Between the watersheds are ridgelines, and crossing a ridgeline reverses the direction of flow on one or more edges. Within a watershed, gradient descent can move the network toward the local optimum for the watershed, but deteriorates if the local optimum is on the ridgeline. 
\begin{figure}
\begin{center}
\includegraphics[width=0.65\textwidth]{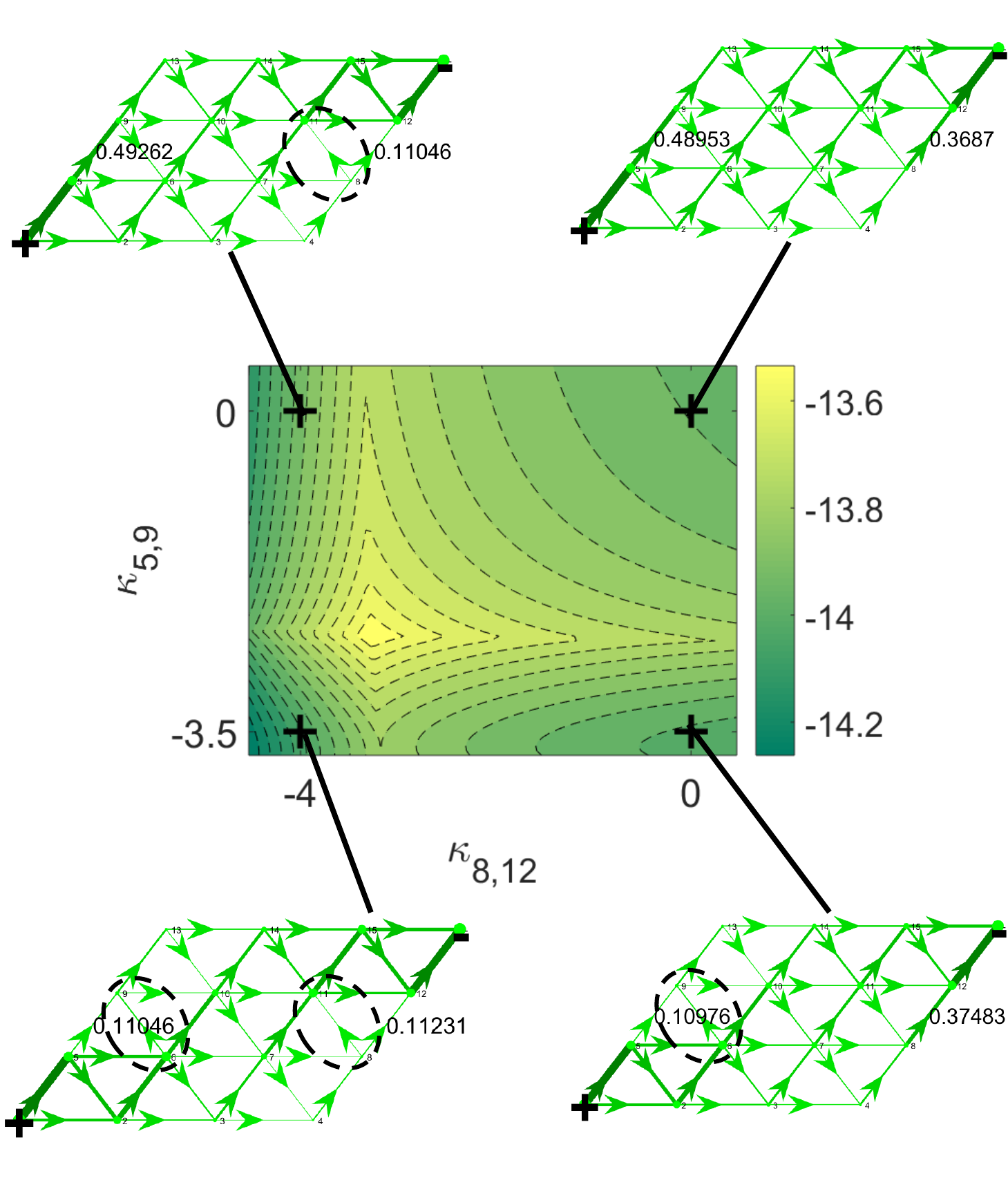}
\caption{Landscape of the mixing entropy. Original network is shown at top right. We systematically perturb the two conductances: $\kappa_{8,12}$ and $\kappa_{5,9}$. At some critical perturbation of the conductances, the flows $q_{8,11}$ and $q_{6,9}$ are respectively reversed (reversed flows are circled in the network plots). The landscape of mixing entropy (middle) has ridge lines where the flow is reversed.} \label{fig:mixinglandscape}
\end{center}
\end{figure}

Even differentiable functions like dissipation produce landscapes with many local optima; accordingly, in \cite{katifori2010damage} simulated annealing and diffusive rearrangements of conductances were implemented to prevent networks from being trapped at unfavorable local optima. We follow a similar approach, by augmenting a gradient-based search that is constrained to remain within a single watershed, with a perturbation method that is designed to provide the network with alternate routes to explore, and by intentional search over adjacent watersheds. We describe the three parts and their integration below.

\subsection{Part 1 of the optimization algorithm: Gradient-based local search}

We perform a gradient-based search, via MATLAB's implementation of the interior-point method in fmincon. Only edges with conductances larger than $10^{-4}$ at the initial state are optimized with smaller conductances treated as constant so that the dimension of the search space is not unnecessarily large. To ensure the search is not challenged to cross the ridges that divide different flow topologies, we enforce the sign of flow in each edge via a set of non-linear constraints on the conductances. Although fmincon is capable of calculating the derivative of $\Theta$ numerically within a watershed, we accelerate the algorithm by computing the gradient analytically using Lagrange multipliers to encode all of the relationships between conductance, flow, transition probabilities and mixing entropy:

We rewrite the array $\kappa_{ij}$ as $|\mathcal{E}|$-entry vector.  The function $\Theta$ (from Eqn. \ref{eq:optimfn}) that we are seeking to optimize is built up from $\kappa_{ij}$ via a chain of dependencies
\begin{equation}
\kappa_{ij}\longmapsto p_{i}\Longmapsto f_{i}\Longmapsto T_{ij}\longmapsto P_{ij}\Longmapsto\tilde{q}_{ij}\longmapsto N_{j}\Longmapsto H. \label{eq:composition}
\end{equation}
Where a single arrow $\longmapsto$ represents a function of the immediately
preceding variable and $\Longmapsto$ represents a function of more
than one of the variables to the left. All of the relationships between variables are described in Section \ref{sec:entropy}. Although it is possible to carry derivatives through this list of compositions, the overhead from isolating and using several derivatives of arrayed functions with respect to arrayed variables, makes the gradient computation forbiddingly slow \cite{PhDThesis}. Instead we follow a similar approach to \cite{chang2018minimal} and use Lagrange multipliers to enforce all of the functional relationships that are embodied in Eq. (\ref{eq:composition}). Eq. (\ref{eq:composition}) then becomes a road-map for the order in which we solve for each of the Lagrange multipliers in our system. The constrained version of Eq. (\ref{eq:optimfn}), omitting the dissipation, becomes:
\begin{eqnarray}
\Theta  = & \sum_{i\in\mathcal{N}}f_{i}\sum_{j:\tilde{q}_{ji}>0}\frac{\tilde{q}_{ji}}{N_{i}}\log\left(\frac{\tilde{q}_{ji}}{N_{i}}\right)-\sum_{i\in\mathcal{N}}\alpha_{i}\left(N_{i}-\sum_{j:\tilde{q}_{ji}>0}\tilde{q}_{ji}\right)-\sum_{i,j\in\mathcal{N}}\gamma_{ij}\left(\tilde{q}_{ij}-f_{i}P_{ij}\right) \nonumber\\
 & -  \sum_{ij\in\mathcal{N}}\mu_{ij}\left(\delta_{ij}-\left(P_{ij}-\sum_{l\in\mathcal{N}}T_{il}P_{lj}\right)\right)-\sum_{i}\sum_{j\in n(i)}\lambda_{ij}\left(T_{ij}-\frac{q_{ij}\mathbf{1}_{q_{ij}>0}}{f_{i}}\right)\\
 & -  \sum_{i\in\mathcal{N}}\beta_{i}\left(f_{i}-\sum_{j\in n(i)}q_{ij}\mathbf{1}_{q_{ij}>0}+\left|Q_{i}\right|\mathbf{1}_{Q_{i}<0}\right)-\sum_{i}\nu_{i}\left(Q_{i}-\sum_{j\in n(i)}\kappa_{ij}(p_{i}-p_{j})\right). \nonumber
\end{eqnarray}

For our gradient descent, we make use of the derivative:
\begin{eqnarray}
\frac{\partial\Theta}{\partial\kappa_{ab}}  = & \lambda_{ab}\frac{\mathbf{1}_{q_{ab}>0}\left(p_{a}-p_{b}\right)}{f_{a}}+\lambda_{ba}\frac{\mathbf{1}_{q_{ba}>0}(p_{b}-p_{a})}{f_{b}}(\beta_{a}(p_{a}-p_{b})\mathbf{1}_{q_{ab}>0}+\beta_{b}(p_{b}-p_{a})\mathbf{1}_{q_{ba}>0})\nonumber \\
 & + (\nu_{a}-\nu_{b})(p_{a}-p_{b})-c(p_a-p_b)^2.
\end{eqnarray}
In which we have made use of the derivatives compiled in Appendix \ref{sec:appendix_derivs} to calculate the derivative of the dissipation.

Our working algorithm uses the above gradients, along with two further transformations. First, we require that all conductances be non-negative. We ensure this by representing our network in terms of {\it log conductances}, defined by: $\kappa_{ij} = \exp(\tilde{\kappa}_{ij})$. Additionally we want to ensure that the total material investment in the network remains constant; i.e. to ensure $\sum \kappa_{ij}^\gamma = C$. In \cite{chang2018minimal} this constraint was added via an additional Lagrange multiplier, but this method guaranteed that the constraint is satisfied only at leading order in the step size. Hence, here we simply rescale the conductances: $\kappa_{ij}\mapsto\left(\frac{C}{\sum_{(i,j)}\kappa_{ij}^{\gamma}}\right)^{\frac{1}{\gamma}}\kappa_{ij}$ after each perturbation. Both transformations need to be considered when calculating the derivatives. For the rescaling we get:
\begin{eqnarray}
\frac{\partial}{\partial\kappa_{ij}}\left(\frac{C^{1/\gamma}\kappa_{ab}}{\sum\kappa_{cd}^{\gamma}}\right) & = & \frac{C^{1/\gamma}}{(\sum\kappa_{cd}^{\gamma})^{1/\gamma}}\left(\delta_{(ab),(ij)}-\frac{\kappa_{ab}\kappa_{ij}^{\gamma-1}}{\sum \kappa_{cd}^\gamma} \right)
\end{eqnarray}
To turn derivatives with respect to $\kappa_{ij}$ into derivatives with respect to $\tilde{\kappa}_{ij}$ we pre-multiply them by $\frac{\partial \kappa_{ab}}{\partial \tilde{\kappa}_{ij}} = \delta_{ai}\delta_{bj}\kappa_{ij}$. 

\subsection{Part 2 of the optimization algorithm: Redistributing material}

Similar to dissipation-minimizing networks \cite{katifori2010damage} our optimization algorithm has many local optima in which source and sink are sparsely connected. To find the true global optimum, our algorithm includes a step for redistributing material within the network, in a way that presents the algorithm with a range of paths of different lengths between source and sink. However, although \cite{katifori2010damage} previously redistributed material by diffusing it on the graph, we found this method tends to short circuit the network by introducing much shorter paths between source and sink. The appearance of these paths is catastrophic for the optimization algorithm, since they are attracting local optima but far from the global optima \cite{PhDThesis}. Since our algorithm does not send conductances exactly to zero, we define a threshold conductance $\kappa_c$, and say that an edge (as well as the vertices that it connects) is in the support of the network if its conductance exceeds $\kappa_c$. In practice we found that a value $\kappa_c=2\times 10^{-2}$ worked for all of the simulations shown in this paper.

We redistribute material using a \emph{network growth step}, which adds spurs of material from the network's support. Our algorithm takes the form of a set of operators: $\text{Grow}_{\text{up-right}}(\kappa_{ij})$
where ``up-right'' in the subscript
can be replaced with the ``up-left'', ``down-left'' or ``down-right'' to
denote the direction in which material is added. Growth in the up-right direction adds edges that link nodes in the support to nodes not in the support that are upwards and right of them. Each step of the growth algorithm concatenates growth in two non-parallel directions. In practice we did not find it necessary to include right or left growth. We will describe the up-right growth step: other growth steps can be derived from this step by symmetry.

\begin{enumerate}
\item First locate up-right edges in the triangle grid for which the
top-right node is outside of the support of the network and the bottom-left node is inside the support.
\item Add positive conductance to each of these edges to form a new network $\kappa_{\text{up-right}}$ (Fig. \ref{fig:grow_triangle_2}). 
Each new edge is assigned conductance equal to the average conductance of the edges from the support adjacent to its bottom-left node in the set of bottom-left nodes identified in 1.
\item In $\kappa_{\text{up-right}}$ locate the nodes which are top right nodes of
edges in the support. Call this set the top-right nodes (Fig. \ref{fig:grow_triangle_3}).
\item For every top-right node $i$, if $i$ is the apex of a triangle whose base and up-right edge lie in the support of $\kappa_{\text{up-right}}$, complete this triangle with an up-left edge whose conductance is the arithmetic mean of the other two edges. For every top-right node $i$ that is the left vertex of an inverted triangle, whose up-left and up-right edges lie in the support of $\kappa_{\text{up-right}}$, complete the triangle with a horizontal edge whose conductance is the arithmetic mean of the other two edges (Fig \ref{fig:grow_triangle_final}).
\end{enumerate}
\begin{figure}
\begin{center}
\subfloat[]{\label{fig:grow_triangle_1} \includegraphics[width=0.18\textwidth]{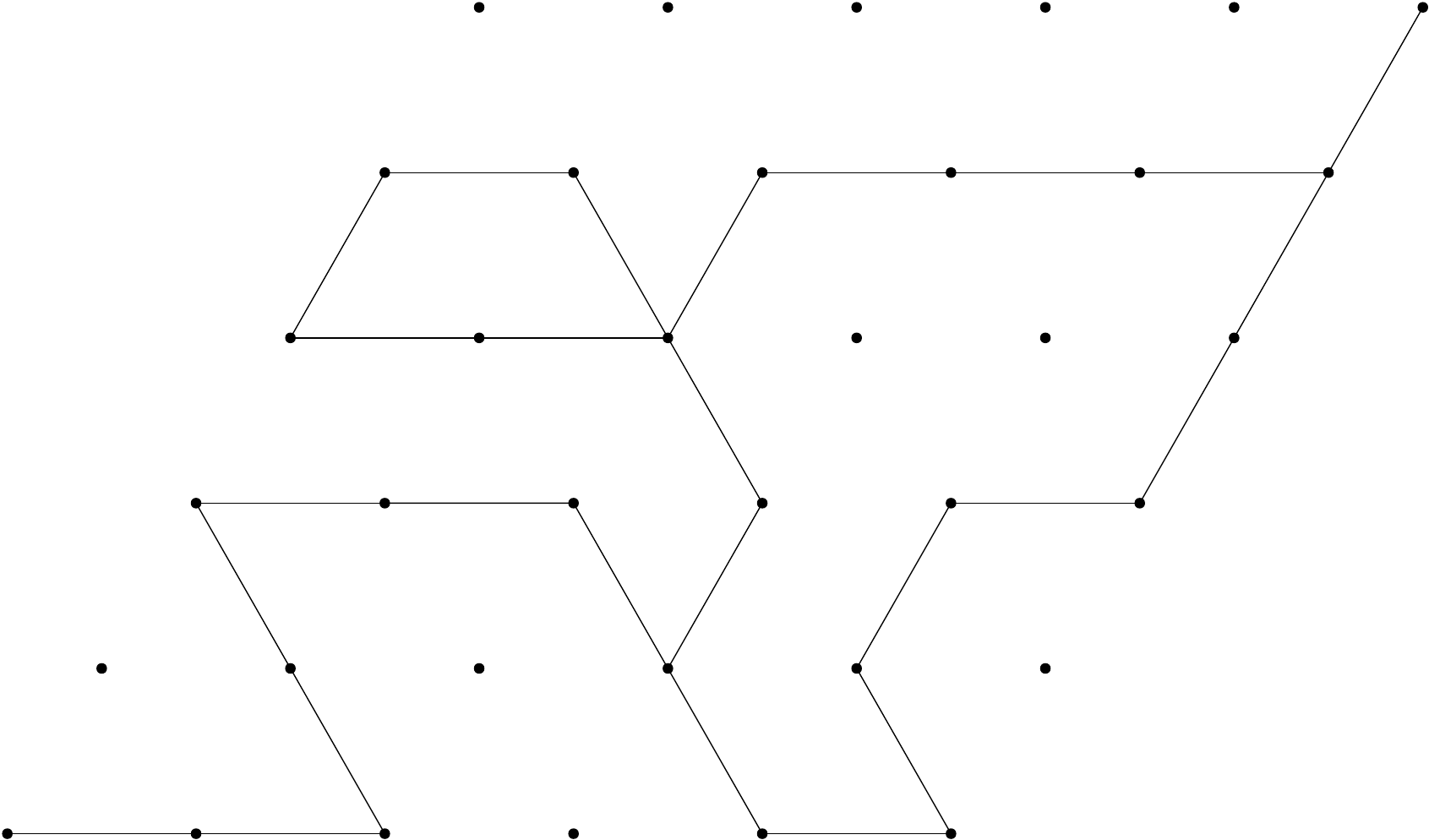}}
\subfloat[]{\label{fig:grow_triangle_2}\includegraphics[width=0.18\textwidth]{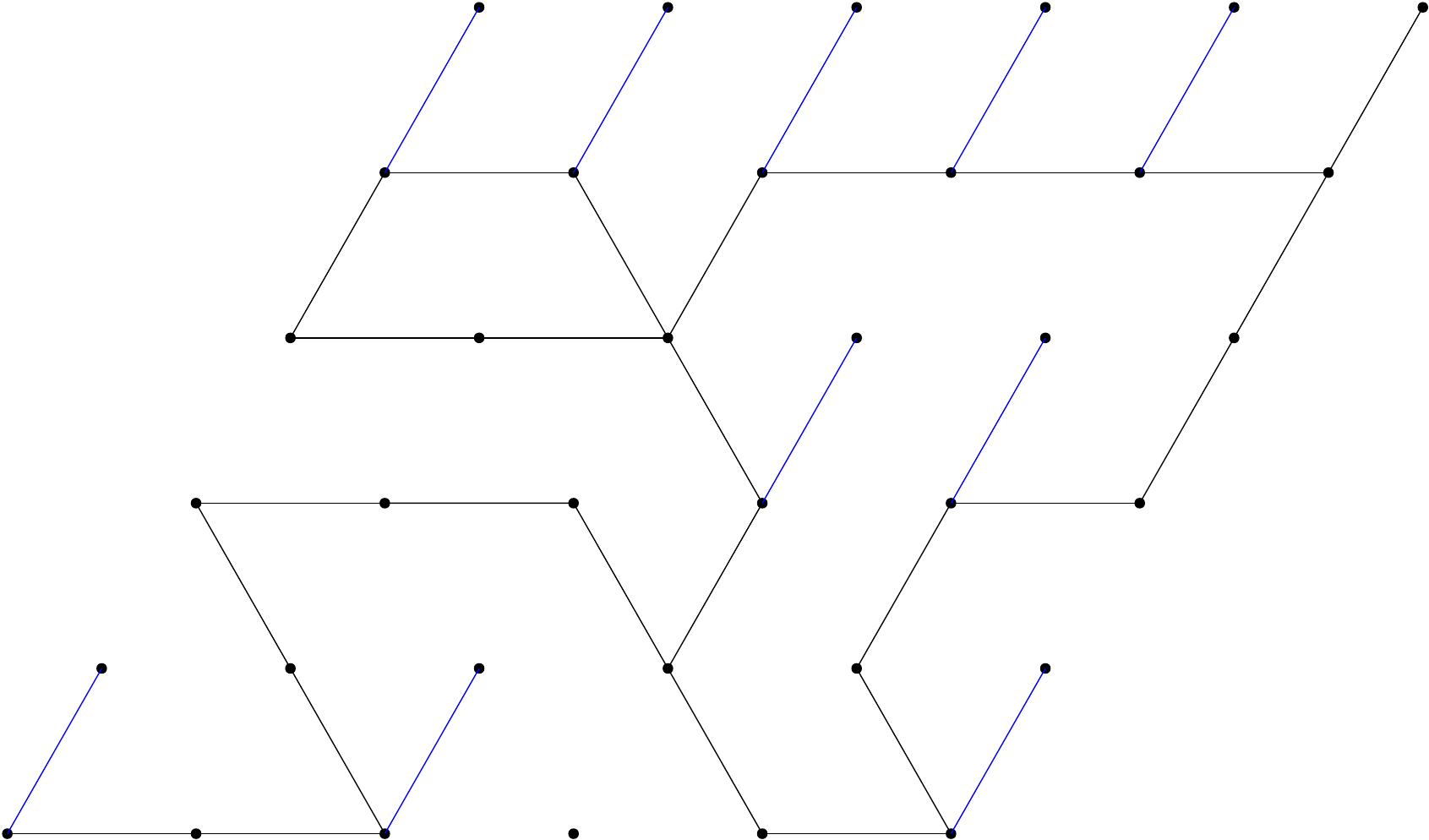}}
\subfloat[]{\label{fig:grow_triangle_3}\includegraphics[width=0.18\textwidth]{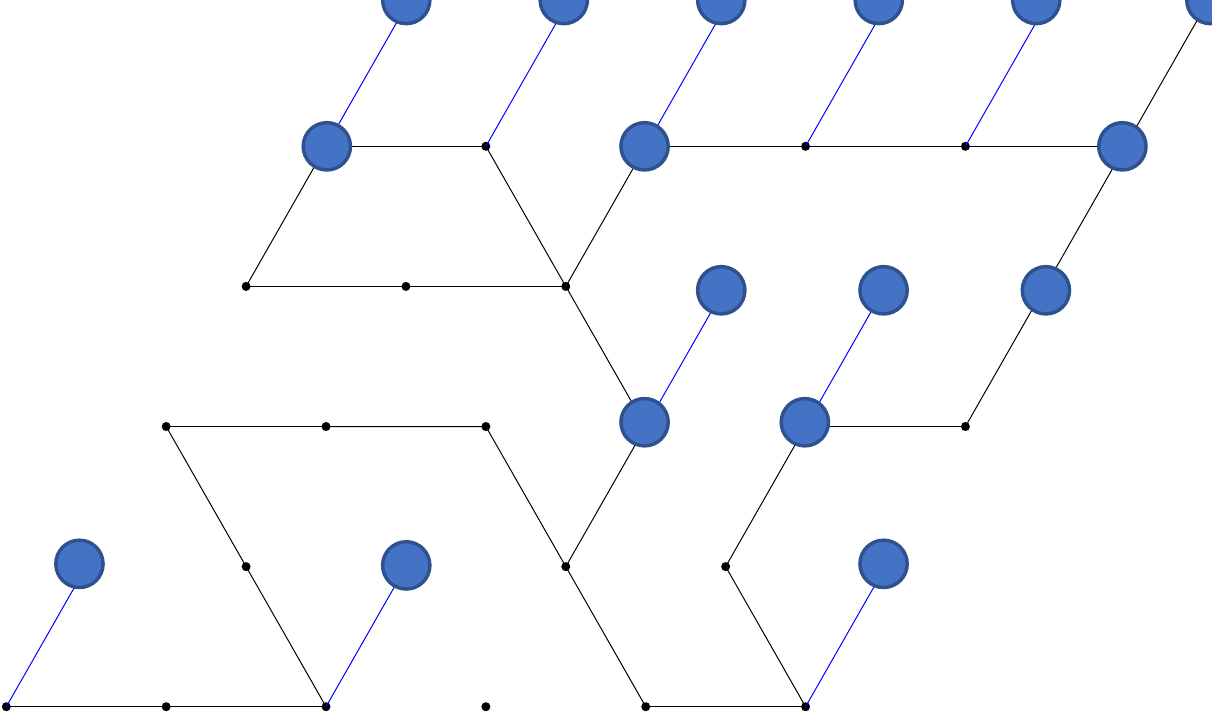}}
\subfloat[]{\label{fig:grow_triangle_final}\includegraphics[width=0.18\textwidth]{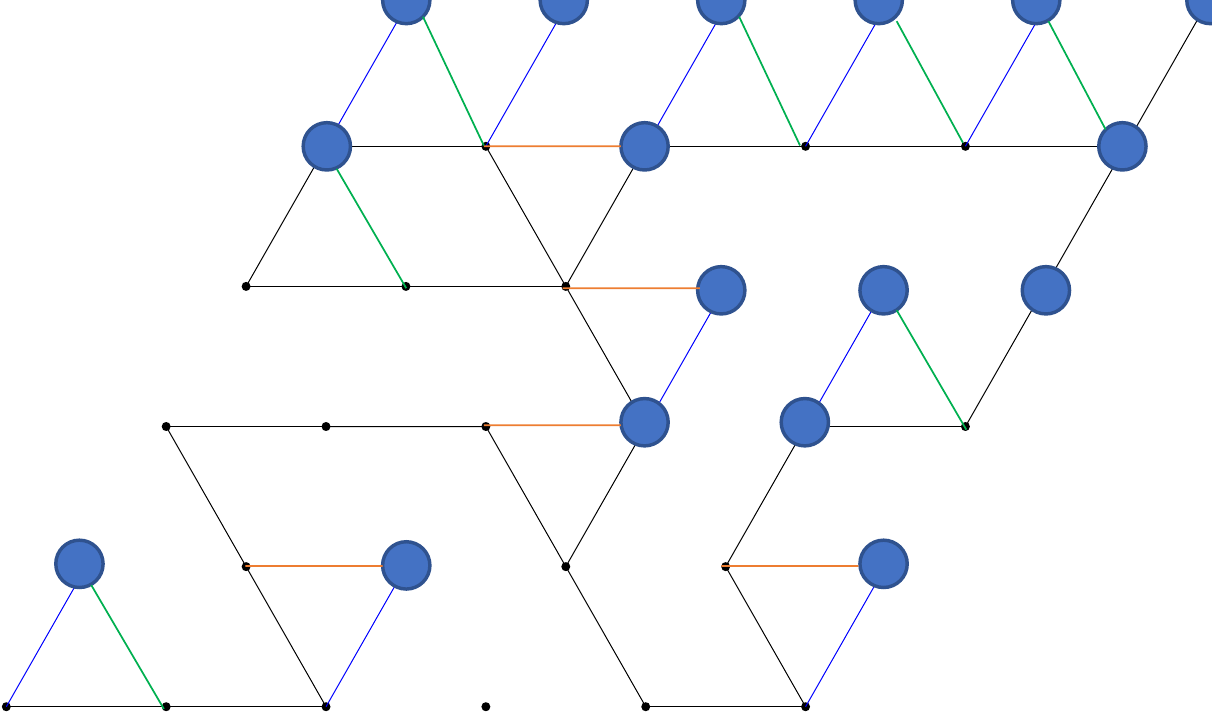}}
\subfloat[]{\label{fig:twosteps2}\includegraphics[width=0.18\textwidth]{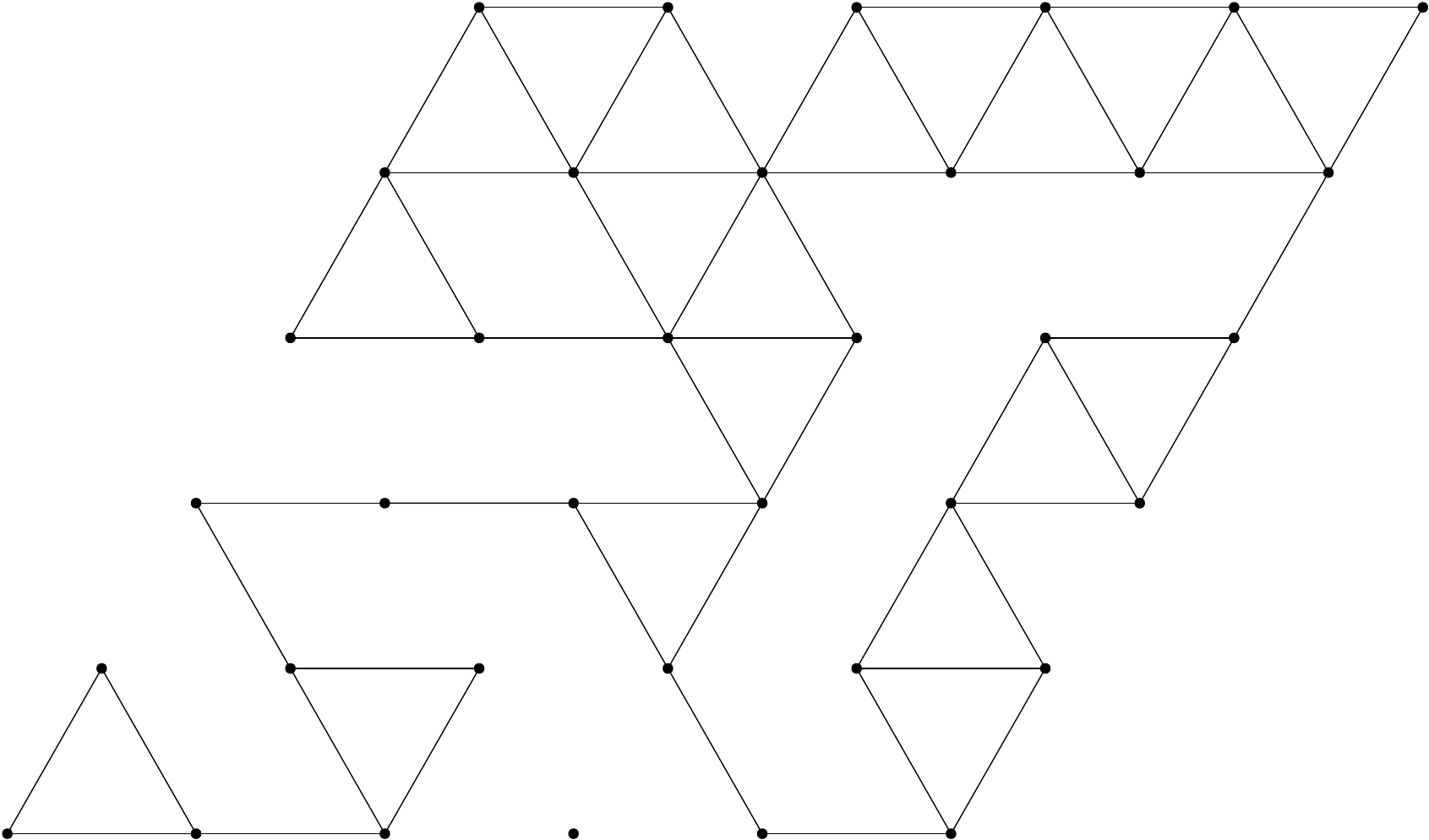}}
\caption{Sequence of steps performed to construct $\kappa_{\text{up-right}}$ from an initial network $\kappa$ shown in (a). (b) blue edges are added up-rightwards from a (bottom) node in the support of the network to a (top) node outside of the support of the network. (c) We identify \textbf{top-right nodes} (blue circles) whose bottom left edges have non-zero conductance following the first step. (d) Add leftwards (orange) and down-rightwards (green) edges if they complete triangles with edges in the current support. (e) Our optimization scheme applies the up-right growth algorithm twice.}
\end{center}
\end{figure}

\subsection{Part 3 of the optimization algorithm: changing flow directions}

Part 1 of our optimization algorithm can reliably locate local optima while respecting the directions of flow on every edge (i.e. the \textit{flow topology}). To find the true global optimum we search systematically over adjacent topologies. To do this, we take one edge within the network, and find the smallest increase and decrease in the conductance of the edge that will change the direction of flow in at least one other edge within the network. To find the smallest change in conductance, we use the Sherman-Morrison formula \cite{sherman1950adjustment}, which allows us to calculate an explicit expression for the conductance change necessary to reverse the flow in any edge of the network (Eq. \ref{eq:flowreversal}). Given a causal edge $(a,b)$, Eq. \ref{eq:flowreversal} enables us to compute a set of perturbations $t_{abuv}$ to $\kappa_{ab}$ to reverse the flow on any edge $(u,v)$. We filter these perturbations to keep only perturbations in which $\kappa_{ab}$ is not allowed to become too small (in practice a threshold of $10^{-3}$ gives good results), to prevent this part of the algorithm 
getting stuck engineering and then re-engineering flow reversals on
edges that already have very low conductance. This method was used to find the set of flow changes shown in Fig. \ref{fig:mixinglandscape}.

\subsection{Synthesis of parts, initialization and termination} \label{sec:synthesis}

We initialize the algorithm by assigning each edge within the ambient network an $U(0,1)$ conductance, and then scaling all conductances to ensure $\sum \kappa_{ij}^\gamma=C$. A single step of the algorithm consists of running all three of its parts sequentially. Part 1 locates a locally optimal network that respects the flow directions given to it, while the random choice of growth directions in Part 2 and of causal conductances in Part 3 stochastically alters the topology of the network. We compare the local optima arrived at the end of consecutive Part 1's; if the new local optimum has a lower value of $\Theta$ than the old, we keep it, otherwise we revert to the old optimum.

Our descent step uses the MATLAB optimization function fmincon using the interior-point algorithm
with 1000 max iterations, with flow directions constrained on all edges with non-negligible conductance (see below), and with analytically computed derivative. Our growth step requires first picking a single direction in which to grow the network: up-left, up-right, down-left or down-right. To ensure that every direction is sampled, we sequentially step through a permutation of all 4 growth directions, choosing a new permutation every 4 steps.

We then count the number of times that during a successful step $\Theta$ decreased by less than $10^{-2}$. When this count reaches $4$ we terminate the algorithm. Otherwise we allow the algorithm to run for 50 iterations. Usually the stopping criterion is reached in fewer than 15 steps. We tested that our algorithm reliably (i.e. in more than half of runs) located the theoretically obtained optimal network when $c=0.05$. Our algorithm constrains all conductances to be be positive, through the use of the coordinate transformation $\kappa_{ij} = \exp(\tilde{\kappa}_{ij})$. In practice, the local optima located by our algorithm use only a subset of the edges in the ambient network. We disregard edges with small conductances (in practice any edges with conductance less than $2\times 10^{-4}$): specifically the directions of flow on these edges are not considered when constraining flow directions in part 1 or when determining the perturbations that cause flows to switch in part 3. 

After the algorithm terminates we perform a final filtering step to deal with the fact that our gradient search is somewhat slow at removing edges from the network or redistributing material between high conductance edges. To filter, we set all edges with conductance $\leq 10^{-3}$ to $10^{-9}$, re-scale all edges so that the material cost of the network stays the same and then run our gradient-search with 10000 max iterations.

Most of our simulations involve sweeps over $c-$values (see Section \ref{sec:results}), typically involving 100-200 replicate networks whose $c$ values are close enough that we expect them to be topologically equivalent. We can further boost coverage since any local optimum, $\hat{\kappa}$ discovered by our algorithm at $c=\hat{c}$ can be compared with local optima for different values of $c$ by tracing the line: $\hat{\theta} = -H(\hat{\kappa})+cD(\hat{\kappa})$. We form the envelope of these straight lines (Fig. \ref{fig:triangle_1src1snk_5_theta}). At any value of $c$, we identify the network that produces the straight line on which $(c,\Theta(c))$ lies as the {\bf global optimal} network for that value of $c$. 

\section{Results from numerical optimization}
\label{sec:results}

\subsection{Optimal networks are paths for small values of $\gamma$}

We first studied the effect of fixing the value of $c$ and constructing optimal networks over a range of values for $\gamma$. We found that at each value of $c$ the number of loops in the path increased with $\gamma$ (cf. dissipation minimizing networks, which form loops only when $\gamma>1$ \cite{durand2007structure,bohn2007structure}). At physiologically relevant values of $\gamma$ ($\gamma\approx 0.45-0.5$) the globally optimal networks are simple loopless paths linking source to sink.
\begin{figure}
\begin{center}
\includegraphics[width=0.8\textwidth]{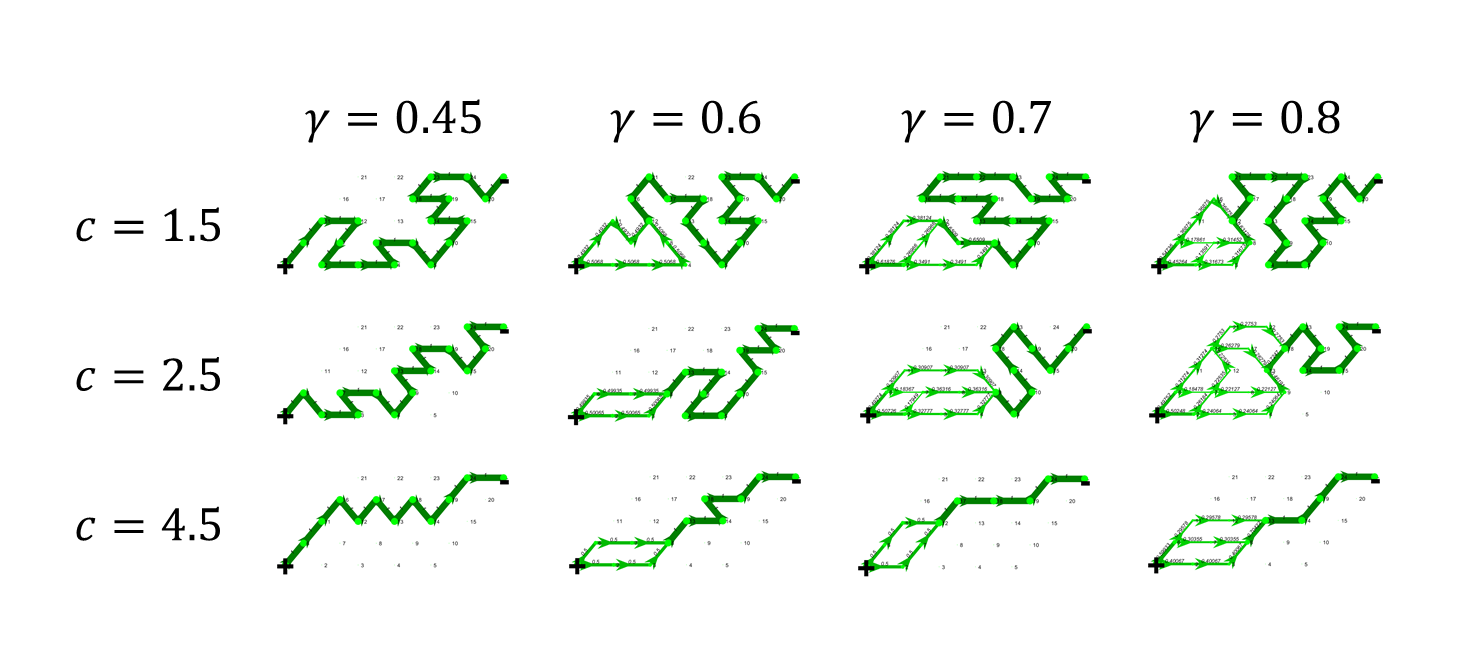}
\caption{Numerically computed optimal networks on a $5\times 5$ triangular grid for three different values of $c$ (rows) and four different values of $\gamma$ (columns). At sufficiently small values of $\gamma$ (in particular for $\gamma=0.45$ at each assayed value of $c$), optimal networks are all paths from source to sink. Increasing $\gamma$ progressively adds loops to the network. Increasing $c$ increases the number of nodes visited by the network.} \label{fig:gammavaries}
\end{center}
\end{figure}

\subsection{Length of optimal networks increases with $c$}

\begin{figure}
\begin{center}
\includegraphics[width=0.22\textwidth]{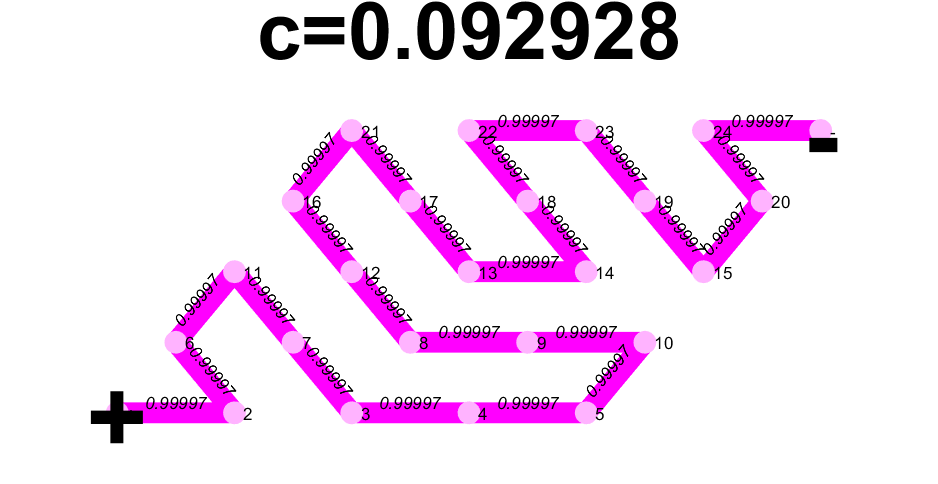}
\includegraphics[width=0.22\textwidth]{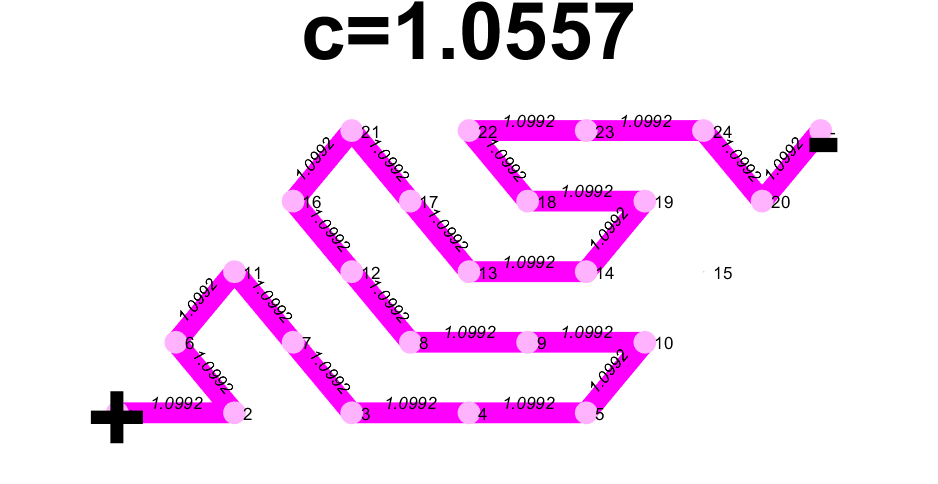}
\includegraphics[width=0.22\textwidth]{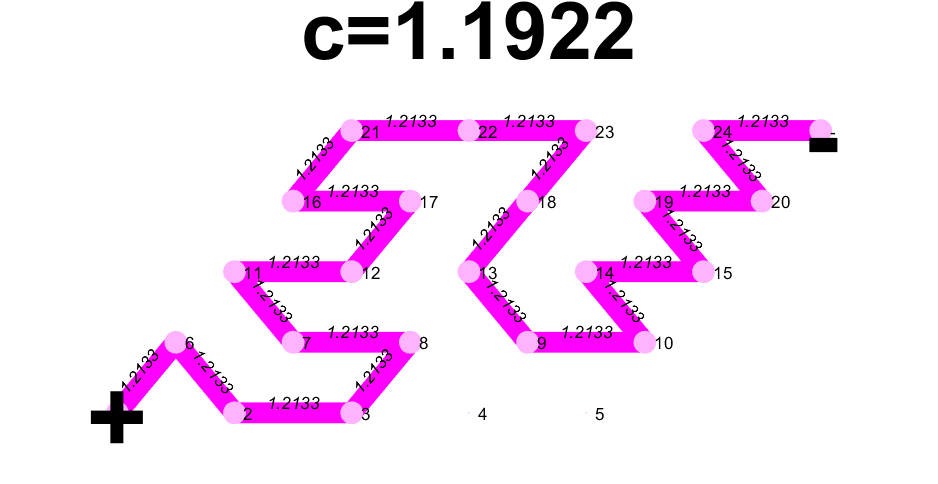}
\includegraphics[width=0.22\textwidth]{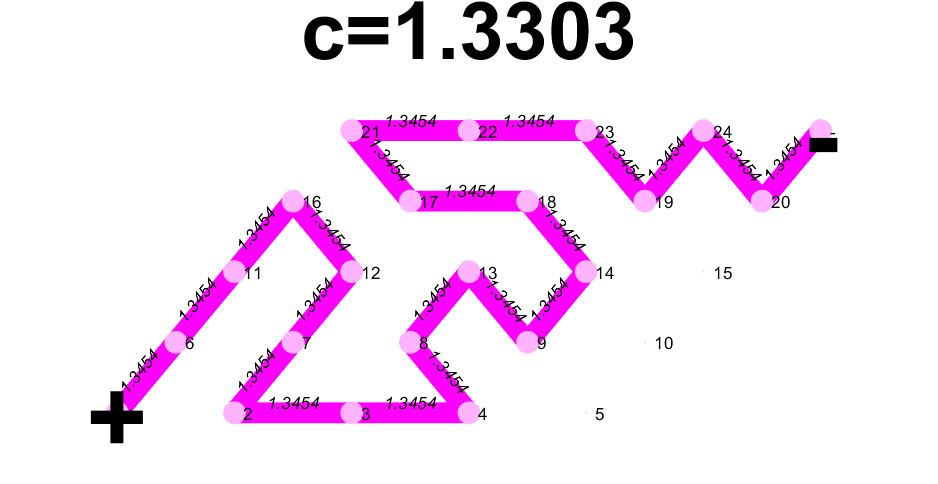}\\
\includegraphics[width=0.22\textwidth]{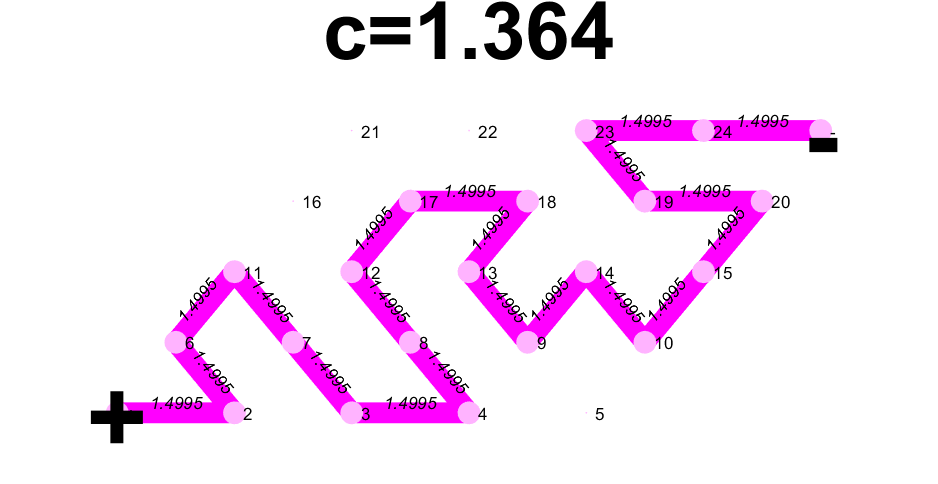}
\includegraphics[width=0.22\textwidth]{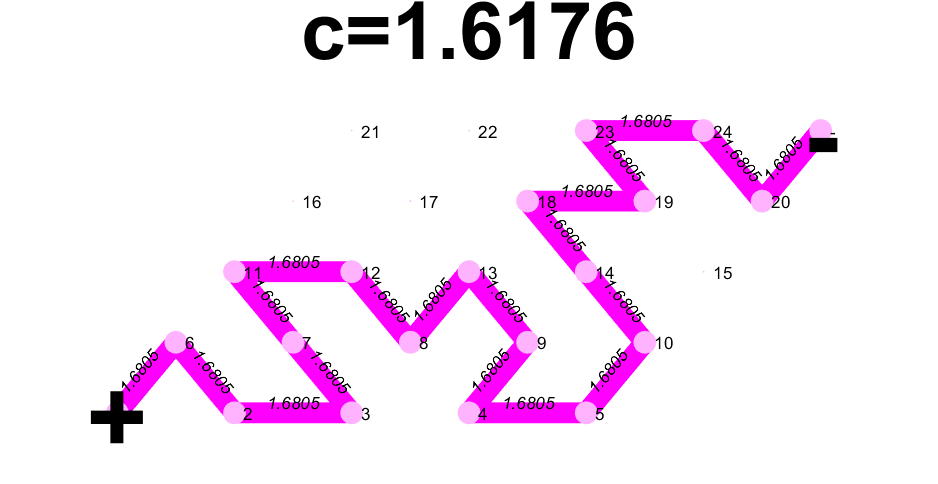}
\includegraphics[width=0.22\textwidth]{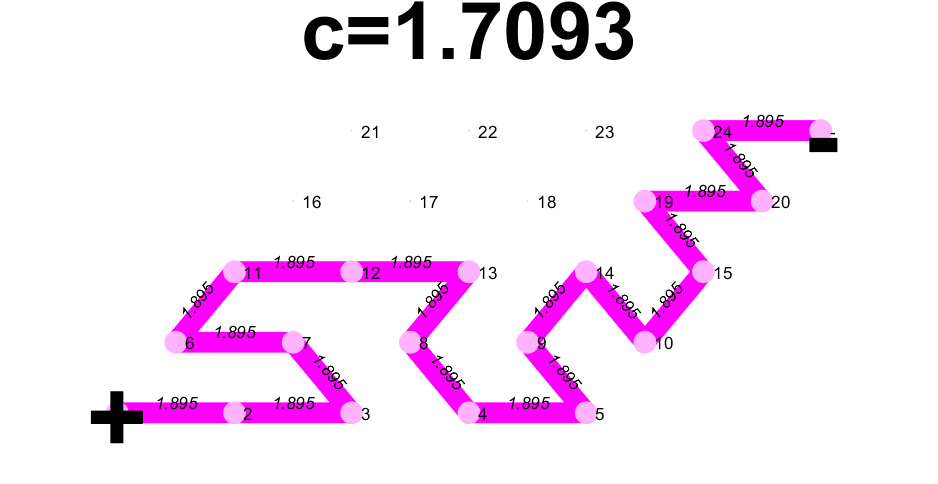}
\includegraphics[width=0.22\textwidth]{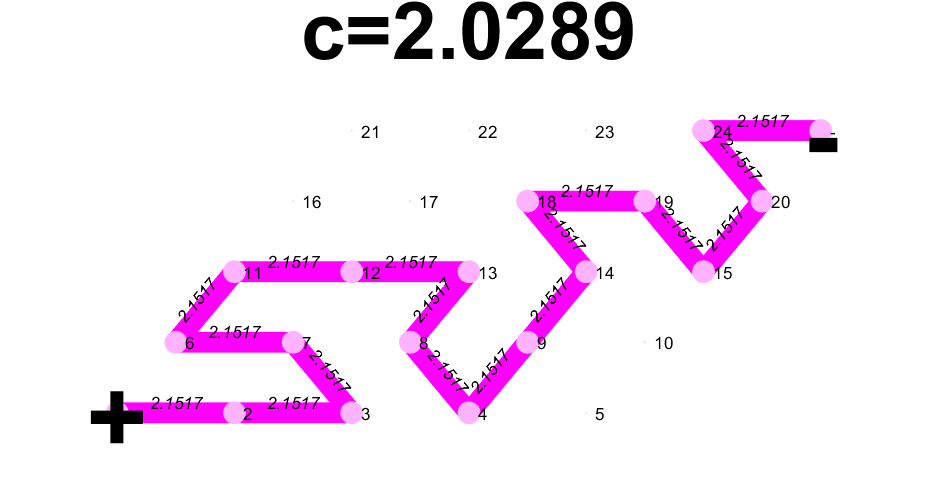}\\
\includegraphics[width=0.22\textwidth]{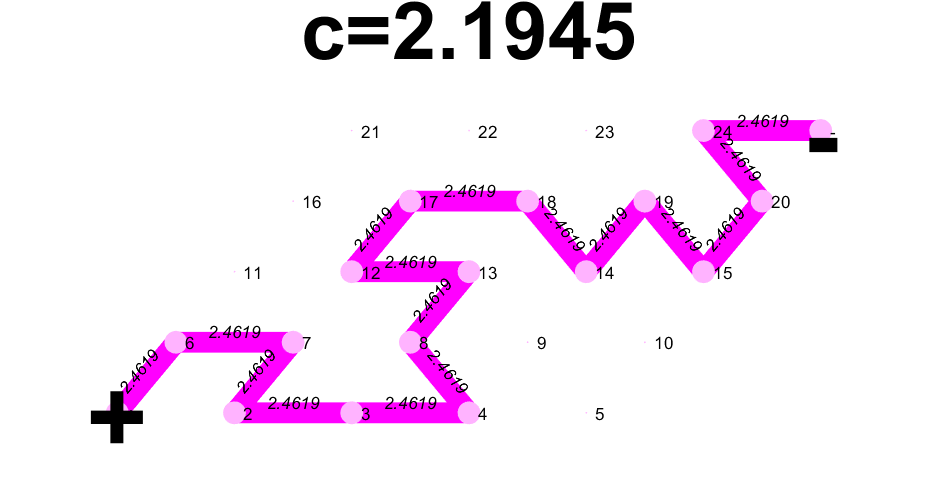}
\includegraphics[width=0.22\textwidth]{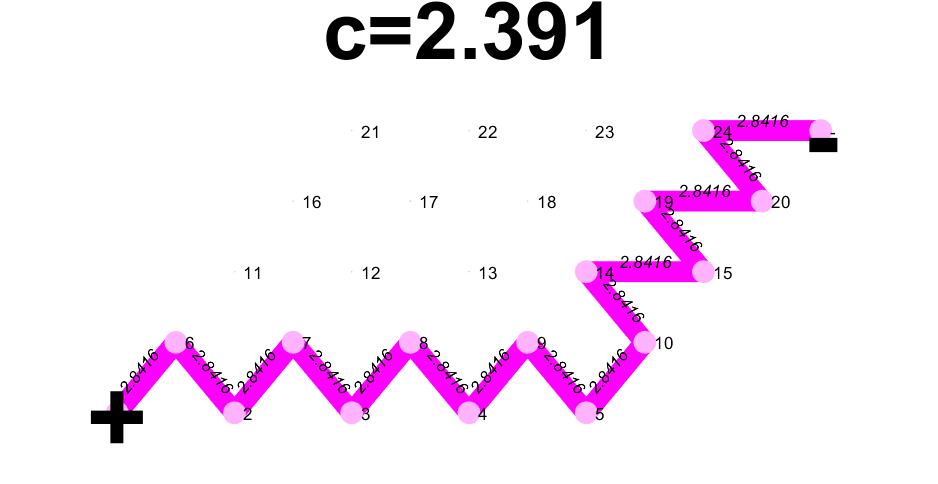}
\includegraphics[width=0.22\textwidth]{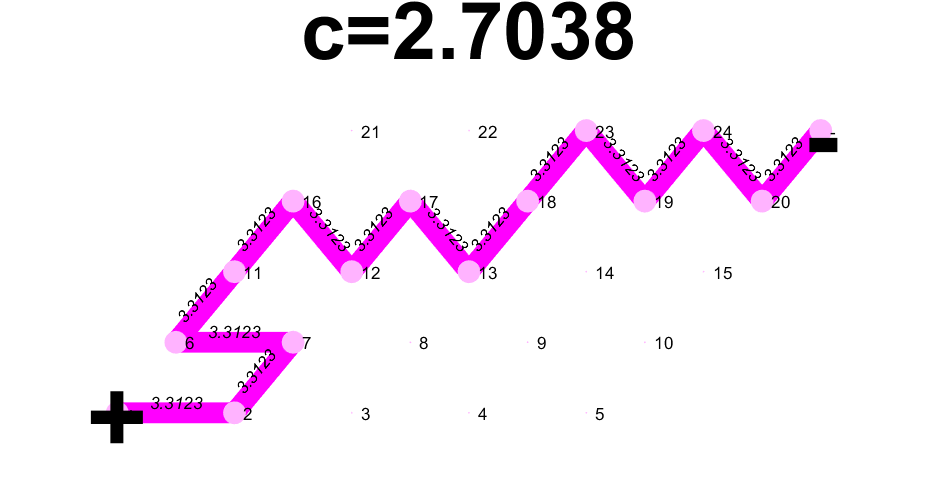}
\includegraphics[width=0.22\textwidth]{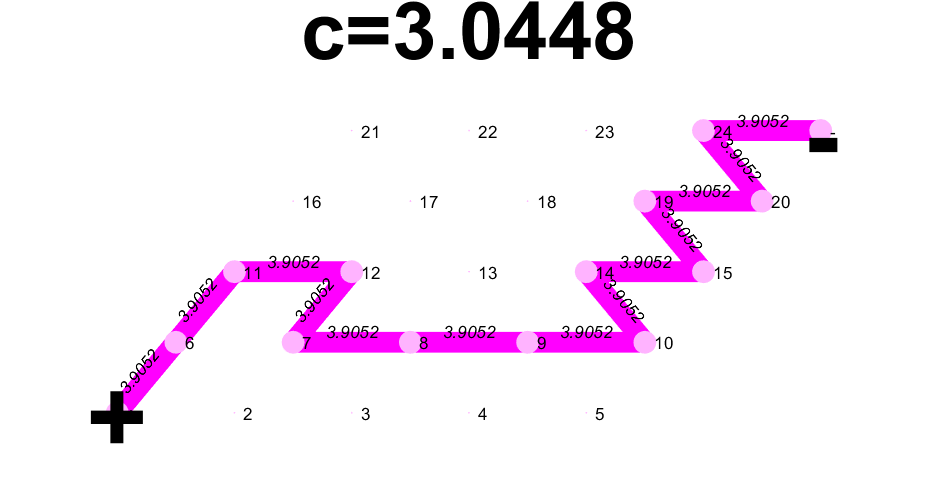}\\
\includegraphics[width=0.22\textwidth]{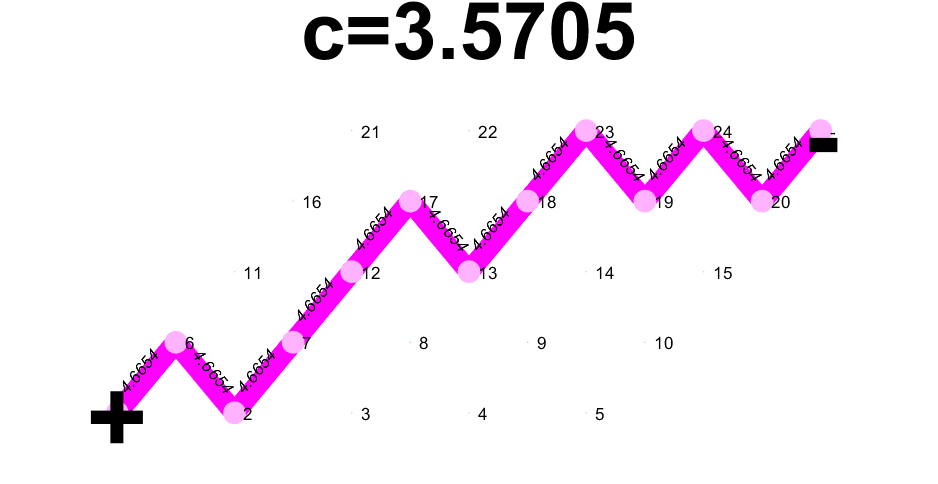}
\includegraphics[width=0.22\textwidth]{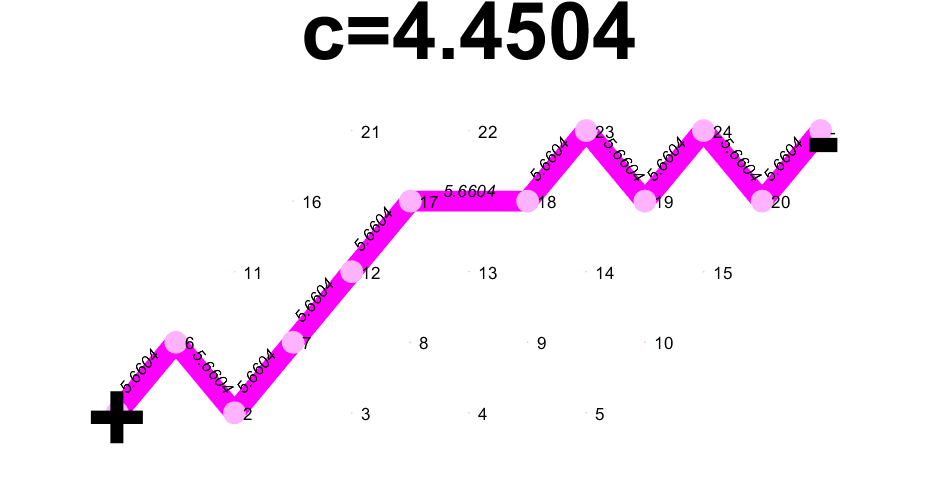}
\includegraphics[width=0.22\textwidth]{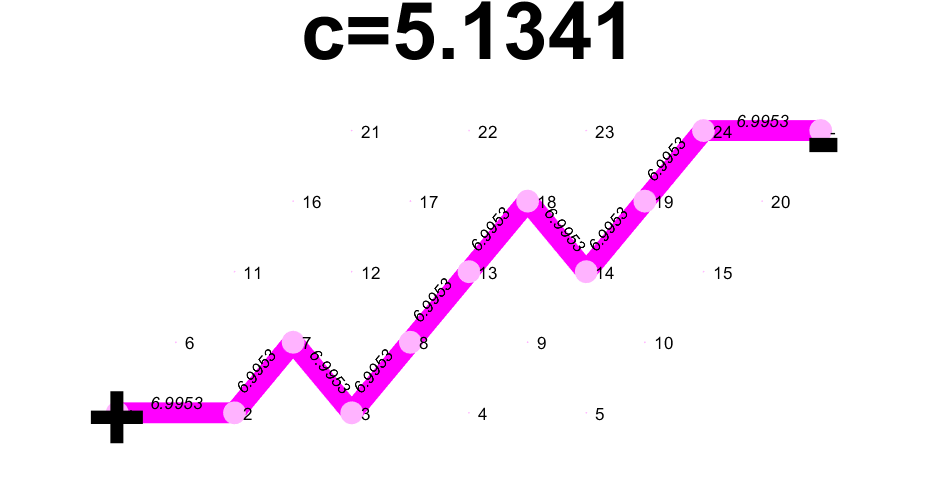}
\includegraphics[width=0.22\textwidth]{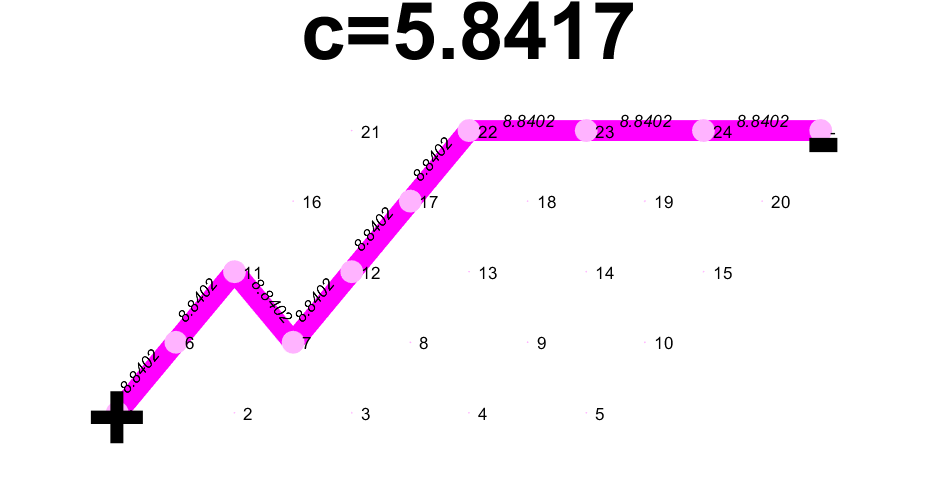}\\
\includegraphics[width=0.22\textwidth]{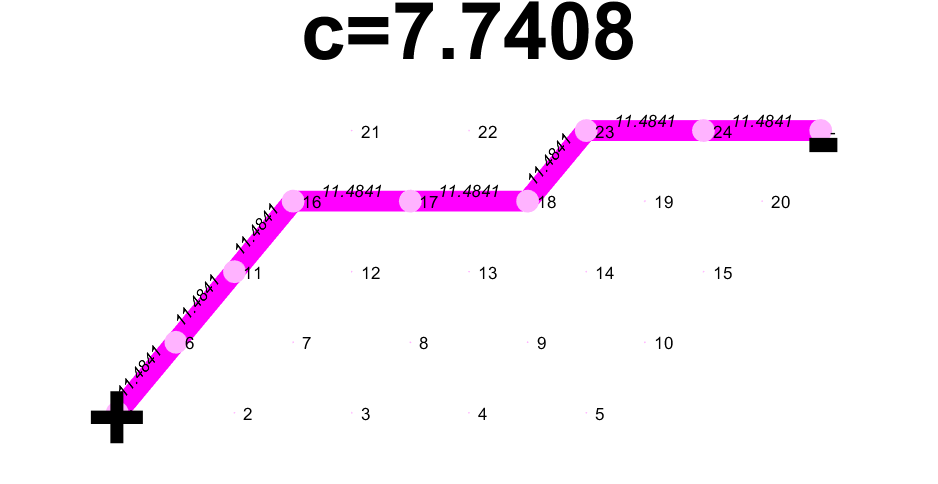}
\caption{Selected optimal networks for $\gamma=0.45$, on a $5\times 5$ triangular grid, for increasing values of $c$ (values given above each panel), are all paths with decreasing lengths. Every possible length of path between $24$ and $9$ edges is obtained. The networks shown are the magenta points in Figure \ref{fig:triangle_1src1snk_5_theta}. } \label{fig:triangle_1src1snk_5_q}
\end{center}
\end{figure}

We noticed that in Fig. \ref{fig:gammavaries} changing $c$ changes the number of edges in the network. To investigate the effect of $c$ more systematically, we performed a numerical sweep of $c$ values, holding $\gamma=0.45$ fixed. All of the optimal networks were simple paths. As $c$ increases, the globally optimal network systematically explored all path lengths from $N^2-1$ (a path that visits every node exactly once), to $2N-1$; (the shortest path linking source to sink). Fig. \ref{fig:triangle_1src1snk_5_theta} shows the complete $\Theta(c)$ trace, including for the networks included in Fig. \ref{fig:triangle_1src1snk_5_q}. The numerically obtained $\Theta(c)$ is piecewise linear, with slope discontinuities at each $c$-value where the length of the optimal network increases by one.
\begin{figure}
\begin{center}
\includegraphics[width=0.6\textwidth]{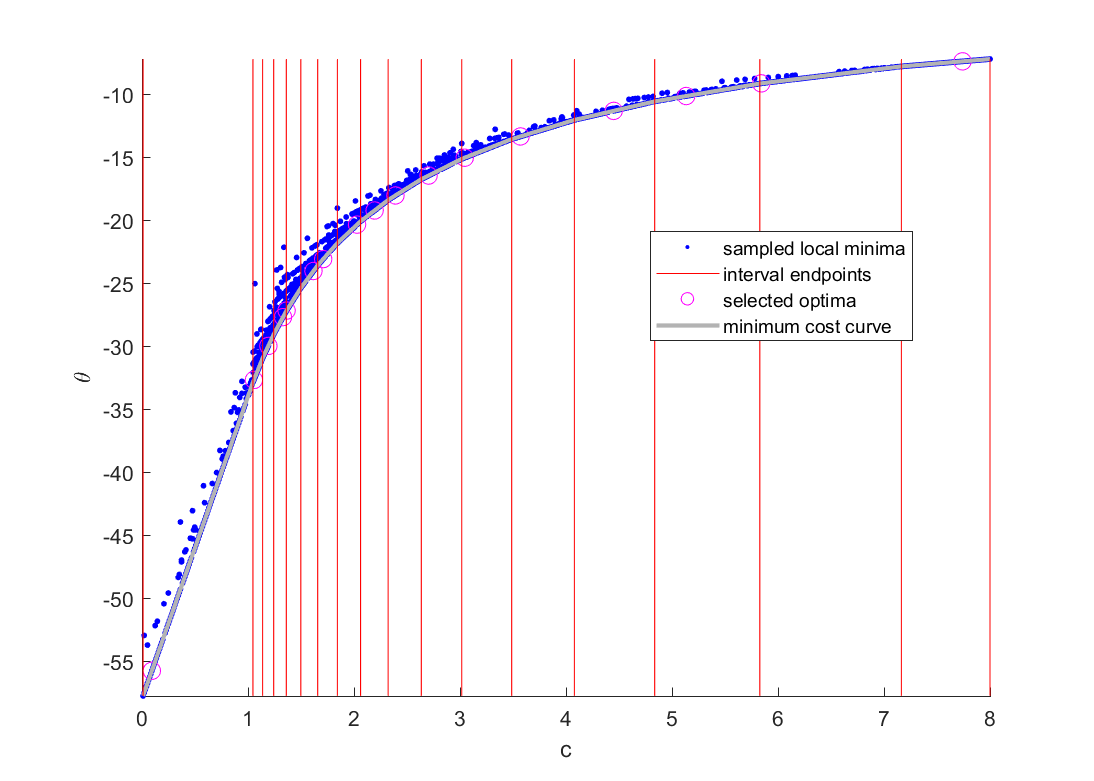}\caption{Optimal mixing-dissipation costs for networks on $5\times 5$ grid, with $\gamma=0.45$ and $K=24$. The $c$-domain is divided into subintervals  $[.01,c_{25,24}]$, $[c_{m+1,m},c_{m,m-1}]$ for $m=9,10,\ldots ,24$ and $[c_{10,9},7]$ ($c_{m,n}$ is defined in Section \ref{sec:theory}), and 100 locally optimal networks are generated within each subinterval. Global optima are derived is described in Section \ref{sec:synthesis}. Selected globally optimal networks (magenta points) are plotted in Fig. \ref{fig:triangle_1src1snk_5_q}.} \label{fig:triangle_1src1snk_5_theta}
\end{center}
\end{figure}

\section{Optimal path networks}
\label{sec:theory}
Our numerical results from Section \ref{sec:results} highlight three properties of optimal networks: that they are simple paths at small and moderate values of $\gamma$, that the path length decreases as $c$ increases, and that for small values of $c$ the simple path visits every vertex in the network. In this section we will rigorously state and prove theorems justifying these properties.

We will first prove separate results for mixing and for dissipation. We will show that the optimal network for mixing (that is, without considering the cost of dissipation) is a {\bf tour} -- a path that visits every node in the network from source to sink. Then given only the flows on a network, we bound its dissipation. Second, we will consider optimization among path (i.e. loopless) networks, showing that as the cost of dissipation is increased, the length of the optimal {\it path} network decreases monotonically in length in steps of 1, from the tour to a geodesic (shortest path). Finally we show that among all networks, over many different values for the dissipation penalty factor $c$, the optimal network is a path for all sufficiently small $\gamma$. We start by introducing a notation for paths of different lengths, assuming that conductances are uniform, i.e. the same on each edge within the path, which is favored for minimizing dissipation.

\begin{definition}
Say that a path from source to sink has length $m$ if it visits exactly $m$ nodes. We use the notation $\tau_m$ to denote any uniform conductance path of length $m$. Further, we call $\tau_{|\mathcal{N}|}$, the path that visits every node in the ambient network, a \emph{tour}. 
\end{definition}

\subsection{The optimal network for mixing is a tour}
\begin{thm}
\label{thm:network_best_NME} Suppose that $q_{ij}$ is a flow network
with node-set $\mathcal{N}$, and $|\mathcal{N}|=n$. Then the maximum possible
total mixing entropy is $\log(n!)$, and this maximum is attained only for a path that visits all $n$ nodes exactly once.
\end{thm}
The intuitive interpretation of this result is that all of the nodes in the network can be ordered by their pressures, $p_i$. Signals from node $i$ can reach node $j$ only if $p_i>p_j$. An optimal mixing configuration is one in which signals from node $i$ reach all downstream $j$ with probability 1, which requires that the downstream network is a path that visits each downstream node in turn.

\begin{proof}
Let $x=\{x_{t}\}_{t\geq0}$ be the random walk on the flow network defined in Section \ref{sec:entropy}. Let
$V_{k}\subset\mathcal{N}$ be the set of nodes $v$ that receive signals from exactly
$k$ nodes: that is, $V_{k}=\left\{ u\in\mathcal{N}:\#\left\{ v:P_{vu}>0\right\} =k \right\} $.

For a subset of nodes $S\subset\mathcal{N}$ we say that $x$ hits $S$, if for some $t\in\mathbb{Z}_{\geq0}$
$x_{t}\in S$. Let $k\geq1$ such that $V_{k}$ is non-empty. Let
$u\not=v\in V_{k}$. A signal $x$ can not visit more than one node in $V_k$ for, if $x$ hits both $v_1$ and $v_2\in V_k$, and WLOG $p_{v_1}>p_{v_2}$, then for any $u\in\mathcal{N}$ with $P_{uv_1}>0$, we must also have $P_{uv_2}>0$. So $\{u:P_{uv_2}>0\}\subsetneq\{u:P_{uv_2}>0\}$, which is impossible since $v_1,v_2\in V_{k}$ implies both of these sets contain $k$ elements.

Let $v\in V_{k}$. Then $H(\mathcal{P}_v)\leq \log(k)$ because there
are exactly $k$ nodes $u$ with $\tilde{q}_{uv}>0$. This inequality, together
with $\sum_{v\in V_{k}}f_{k}\leq1$ gives us the lower bound on $H$:
\begin{equation}
\sum_{i\in\mathcal{N}}f_{i}H(\mathcal{P}_i)=\sum_{k=1}^{n}\sum_{v\in V_{k}}f_{v}H_v \leq \sum_{k=1}^{n}\log(k)=\log(n!).
\end{equation}
We show that the only network with $n$ nodes attaining the maximal entropy
is a path. For the path, labeling the nodes $1,2,\ldots,n$ in the order in which they are visited from source to sink we find for $1\leq i<j\leq n$, $\tilde{q}_{ij}=1$.
Therefore the probability distribution of signals $\mathcal{P}_{i}$
is the uniform distribution on $i$ atoms and has entropy
$H_i = \log(i)$. Hence, $H=\sum_{i=1}^{n}f_{i}H_i=\log(n!)$.

Conversely, if $q_{ij}\neq \tau_{n}$, then there must be at least one $k\in\{1,2,\ldots n\}$ such that
$V_{k}=\emptyset$. In this case, $H$ differs from $\log(n!)$
by at least $\log(k)$.
\end{proof}
\begin{corollary}
\label{prop:Best-for-restricted} Let $q_{ij}$ be a flow network
on nodes $\mathcal{N}$ and $\emptyset\not=\mathcal{F\subset\mathcal{N}}$
with $|\mathcal{F}|=m$. Then $H_{\mathcal{F}}\leq\log(m!)$.
\end{corollary}
%
To deduce the corollary, we treat the probabilities $\mathcal{P}_{\mathcal{F}i}$ in the same fashion as we treated the probabilities $\mathcal{P}_i$ in the proof of Theorem \ref{thm:network_best_NME}. 

\subsection{Dissipation in a network can be bounded given the flows on the network}

The dissipation, $D$, for a network is a function both of its conductances $\kappa_{ij}$ and its flows $q_{ij}$. However, we can bound the dissipation based on the $q_{ij}$, alone, given only the constraint that $\sum_{ij} \kappa_{ij}^\gamma=C$.

\begin{thm} \textbf{Murray's law}. Let $q_{ij}$ be a network of flows, then if $\sum_{ij} \kappa_{ij}^\gamma=C$, the smallest possible dissipation in the network is: $\frac{\left(\sum_{ij}q_{ij}^{\frac{2\gamma}{\gamma+1}}\right)^{1+\frac{1}{\gamma}}}{C^{\frac{1}{\gamma}}}$~. \label{thm:Murray}
\end{thm}

This Theorem is equivalent to Murray's law \cite{chang2018minimal}: it is based on assigning each edge the conductance that minimizes the overall network dissipation.

\begin{proof}
Fixing flows, we minimize the total dissipation over conductances obeying the building constraint $\sum \kappa_{ij}^\gamma=C$. That is we minimize the overall function:
\begin{equation}
D(\kappa_{ij}) = \sum_{ij} \frac{q_{ij}^2}{\kappa_{ij}} + \lambda \left(\sum_{ij} \kappa_{ij}^\gamma - C \right)    
\end{equation}
where the Lagrange multiplier $\lambda$ maximizes the dissipation and we restrict to edges on which $q_{ij}\neq 0$. The minimization of $D$ is performed on a compact set ($\kappa_{ij}\geq 0$ and $\sum_{ij} \kappa_{ij}^\gamma = C$) so the minimum certainly exists. Since $D\to \infty$ whenever $\kappa_{ij}=0$ so the optimal value of $D$ occurs at an interior point within this set. So at the minimum point:
\begin{equation}
    0 = \frac{\partial D}{\partial \kappa_{ab}} = -\frac{q_{ab}^2}{\kappa_{ab}^2} + \lambda \gamma \kappa_{ab}^{\gamma-1}
\end{equation}
solving this equation yields $\kappa_{ab} \propto q_{ab}^{\frac{2}{\gamma+1}}$ (Murray's law), and we find our constant of proportionality by imposing the constraint $\sum_{ab}\kappa_{ab}^\gamma = C$:
\begin{equation}
    \kappa_{ab} = C^{1/\gamma}\frac{q_{ab}^{\frac{2}{\gamma+1}}}{\left(\sum_{ij} q_{ij}^{\frac{2\gamma}{\gamma+1}}\right)^{1/\gamma}}~.
\end{equation}
Substituting for $\kappa_{ab}$ in the dissipation yields the required inequality.
\end{proof}


\subsection{Strong nodes and path-like networks}
Our main results will concern networks that are close to paths; for example path networks that have low conductance excursions adjoined to some of the path nodes. How much do these additions affect the network's mixing entropy? Thinking more generally, we consider networks in which some edges are strong, and others are weak (we will define strong and weak) and bound the contribution of the weak nodes to the network entropy.

\begin{defn}
\label{def:strong-node-def}Let $1,2,\ldots,N$ be a labelling of
the nodes in the network $G$ in decreasing order of total flow $f_{i}$
(that is; $f_{1}\geq f_{2}\geq\cdots\geq f_{N}$). Select $0<\delta<1$
which will be referred to as the \textbf{dominance factor}. Let $k=\min\{i:\delta f_{i}>f_{i+1}\}$. The nodes $1,2,\ldots,k$ are referred to as
the \textbf{strong nodes above dominance factor $\delta$}, denoted
$\mathcal{F}_{\delta}$.
\end{defn}
\begin{thm}
\label{thm:strong-node-theorem}Let $\epsilon>0$ and let $q_{ij}$
be a flow network on an ambient network $G$ with nodes $\mathcal{N}$. Then there exists $\delta>0$, \textbf{depending only on $\boldsymbol{G}$}, such that  $\left|H-H_{\mathcal{F_\delta}}\right|<\epsilon$. In addition, $\delta$ may be chosen such that for each node $i\in\mathcal{F}_\delta$ the nodes adjacent to the largest magnitude in-flow at $i$ and the
largest magnitude out-flow at $i$ are also strong nodes. That is,
if $u,v\in n(i)$ are such that $q_{ui}=\max_{j\in n(i)}q_{ji}$ and
$q_{iv}=\max_{j\in n(i)}q_{ij}$ then $u,v\in\mathcal{F}_\delta$.
\end{thm}
\begin{proof}
Let $\delta>0$ be a dominance factor and, for shorthand, take $\mathcal{F}=\mathcal{F}_{\delta}$
to be the strong nodes in $\mathcal{N}$ over dominance factor $\delta$.
By the triangle inequality, we bound the difference
\begin{equation}
\left|H-H_{\mathcal{F}_\delta}\right|  \leq  \sum_{i\in\mathcal{F}}f_{i}\left|H_i-H_{\mathcal{F}_\delta i}\right|+\left|\sum_{i\not\in\mathcal{F}_\delta}f_{i}H_{i}\right|.
\end{equation}
First we bound the first sum on the right-hand side, a sum over the absolute difference between the different mixing entropies. Let $i\in\mathcal{F}_\delta$ and $j\not\in\mathcal{F}$. Then 
\begin{equation}
\mathcal{P}_{i}(j)  =  \frac{\tilde{q}_{ji}}{\sum_{k\not\in\mathcal{F}_\delta}\tilde{q}_{ki}+\sum_{k\in\mathcal{F}_\delta}\tilde{q}_{ki}}<\frac{f_{i}\delta}{f_{i}}<\delta.
\end{equation}
$\mathcal{P}_{\mathcal{F}_\delta i}$ is obtained by omitting fewer than $N=|\mathcal{N}|$ states from $\mathcal{P}_i$, each with probability less than $\delta$, and then renormalizing to give
a new probability distribution. Since entropy is uniformly continuous
on the simplex $\left\{ \sum_{i=1}^{N}p_{i}=1,p_{i}\geq0\right\} $, we can choose $\delta$ so that $\left|H_{i}-H_{\mathcal{F}i}\right|<\frac{\epsilon}{2N}$
so $\sum_{i\in\mathcal{F}}f_{i}\left|H_{i}-H_{\mathcal{F}i}\right|<\frac{\epsilon}{2}$.

We now bound the magnitude of the second term on the right-hand side.
$H_{\mathcal{F}_\delta i}$ is an entropy of a random variable
taking on less than $N=|\mathcal{N}|$ values. Therefore $\left|H_{\mathcal{F}i}\right|<\log N$.
The total flow through each node, $f_{i}\leq\delta$
for all $i\not\in\mathcal{F}_\delta$. Hence, we have
\begin{equation}
\left|\sum_{i\not\in\mathcal{F}_\delta}f_{i}H_{i}\right|  <  \sum_{i\not\in\mathcal{F}_\delta}\delta\log N<N\,\delta\log N.
\end{equation}
 And so we can choose $\delta$ so that the second term is bounded
by $\frac{\epsilon}{2}$. To complete the proof, note that the magnitudes
of the largest in- and out-flows are $\geq\frac{f_{i}}{\deg i}\geq \frac{f_i}{K}$ where $K$ is the largest degree
of a node in $G$. 
Thus, so long as $\delta<\frac{1}{K}$ the nodes connected to the largest in-
and out-flows of degree $i$ have total flows $>\frac{1}{K}f_{i}>\delta f_{i}$ meaning they are also strong nodes.
\end{proof}
We refer to the network formed by linking the nodes $\mathcal{F}_{\delta}$ up using the edges carrying the maximum inflow and
outflow at each node as the strong network, and re-use notation by using $\mathcal{F}_{\delta}$ to represent the strong network.

\subsection{Optimization of $\Theta$ over paths}

Anticipating our proof in Section \ref{sec:maintheorem} that optimal networks are paths for sufficiently small $\gamma$, we start by restricting our optimization to paths. When restricted to path networks $\Theta(c) = \min_m (-H(\tau_m) + cD(\tau_m))$, i.e. $\Theta$ is the lower envelope of straight lines. We first ask, if $c$ is varied, does the sequence of $\Theta$-minimizing paths always recapitulate Fig. \ref{fig:triangle_1src1snk_5_q}; i.e. start with a tour (at vanishingly small $c$) and end at large, finite $c$ with a geodesic, with the intermediate states being paths whose length increases by 1, at finite and predictable $c$ values. We can rationalize this sequence as follows: For a uniform conductance path of length $m$ each edge carries flow 1, and has conductance $(C/(m-1))^{1/\gamma}$, so the total dissipation is $D(\tau_m) = C^{-1/\gamma}(m-1)^{1+1/\gamma}$, which increases monotonically in $m$. Increasing $c$ increases the relative strength of dissipation to mixing in $\Theta$. Mixing favors tours and, more generally, paths that visit as many nodes as possible, while dissipation favors shorter paths. At each $c$, the optimal path length emerges from the balance of these two competing effects.

 Two paths of different lengths: $\tau_m$ and $\tau_n$, give rise to straight lines $c\mapsto -H(\tau_m) + cD(\tau_m)$ and $c\mapsto -H(\tau_n)+cD(\tau_n)$, with different slopes. Denote the point of intersection between the lines by $c_{m,n}$:
\begin{equation}
K^{-\frac{1}{\gamma}}c_{m,n} = \frac{\log(n!)-\log(m!)}{(n-1)^{1+\frac{1}{\gamma}}-(m-1)^{1+\frac{1}{\gamma}}} = \frac{\log(\Gamma(n+1))-\log(\Gamma(m+1))}{(n-1)^{1+\frac{1}{\gamma}}-(m-1)^{1+\frac{1}{\gamma}}}~.
\end{equation}

\begin{lemma}
The point of intersection $c_{m,n}$ is monotonic decreasing in both $m$ and $n$, for $m, n\geq 2$. \label{lemma:cmnmonotone}
\end{lemma}

\begin{proof}
Let $x_{m}=(m-1)^{1+\frac{1}{\gamma}}$, then: 
\begin{equation}
c_{m,n}  =  \frac{\log\left(\Gamma\left(2+x_{n}^{\frac{\gamma}{\gamma+1}}\right)\right)-\log\left(\Gamma\left(2+x_{m}^{\frac{\gamma}{\gamma+1}}\right)\right)}{x_{n}-x_{m}}.
\end{equation}
So $c_{m,n}$ is the slope of the secant from $\left(x_{m},f(x_{m})\right)$
to $\left(x_{n},f(x_{n})\right)$ where $f(x)=\log\left(\Gamma\left(2+x_{n}^{\frac{\gamma}{\gamma+1}}\right)\right)$.
Since $f$ is an increasing function we need to show it is concave
in order to show that these secant slopes decrease as either $x_{n}$
or $x_{m}$ increases.

Given $f=\log(u(x))$ where $u(x)=\Gamma\left(2+x^{\frac{\gamma}{\gamma+1}}\right)$
we have that 
\begin{equation}
f''(x)  =  \frac{u''(x)u(x)-(u'(x))^{2}}{(u(x))^{2}}~.
\end{equation}
To show that $f$ is concave, we then need that $u''u-(u')^{2}<0$. To compute these derivatives recall $\frac{d}{dx}\Gamma=\Gamma\Psi_{0}$
where $\Psi_{0}$ is the digamma function. The trigamma function $\Psi_{1}$
is defined to be $\Psi_{0}'$, and so $\frac{d^{2}}{dx^{2}}\Gamma  =  (\Psi_{0}^{2}+\Psi_{1})\Gamma$. Pulling these results together, we obtain:
\begin{equation}
u''u-u'^{2} =\gamma\left(\frac{1}{\gamma+1}\right)^{2}x^{-\frac{2}{\gamma+1}}\left(\gamma\Psi_{1}-x^{-\frac{\gamma}{\gamma+1}}\Psi_{0}\right)\Gamma^{2}
\end{equation}
Let $z=x^{\frac{\gamma}{\gamma+1}}+2$. Then $z$ is an increasing function of $x$ and visa versa. Since all of the other multipicative terms in the expression $u''u-u'^2$ are positive we only need to show that $\gamma\Psi_{1}-x^{-\frac{\gamma}{\gamma+1}}\Psi_{0}$ is negative for all $z= x_{n}^{\frac{\gamma}{\gamma+1}}+2=n-1+2\geq3$. We have 
\begin{eqnarray}
\gamma\Psi_{1}-x^{-\frac{\gamma}{\gamma+1}}\Psi_{0}   &\leq&  \Psi_{1}-x^{-\frac{\gamma}{\gamma+1}}\Psi_{0} = \Psi_{1}(z)-\frac{1}{z-2}\Psi_{0}(z)~,\nonumber \\ & \leq & \frac{1}{z}+\frac{1}{z^2}+\frac{1}{z(z-2)}-\frac{\log z}{z-2}~,\nonumber\\
& = & \left(1-\frac{2}{z^2}-\log z\right)\frac{1}{z-2}~.\nonumber
\end{eqnarray}
Here we made use of the inequalities \cite{laforgia2013exponential} $\Psi_0(z)\geq \log z -\frac{1}{z}$, and $\Psi_1(z)\leq \frac{1}{z}+\frac{1}{z^2}$ \cite{laforgia2013exponential} for all $z>0$. The last line is $<0$ for all $z\geq 3$, proving the lemma.
\end{proof}

\begin{thm}
When $\Theta$ is optimized among paths, on a triangular ambient grid with $n$ nodes, and $c$ is increased from 0, the optimal path decreases in length by 1 at predictable values of $c$: $c_{n,n-1}$, $c_{n-1,n-2}$, $c_{n-2,n-3}$ $\ldots$. That is: $\tau_n$ for $c<c_{n,n-1}$, $\tau_{n-1}$ for $c_{n,n-1}<c<c_{n-1,n-2}$, $\tau_{n-2}$ for $c_{n-1,n-2}<c<c_{n-2,n-3}$ and so on.
\end{thm}

\begin{proof}
The theorem follows directly from the monotonicity property proven in Lemma \ref{lemma:cmnmonotone}. We have already shown that the tour $\tau_n$ is the optimal path at $c=0$. As $c$ is increased, the line $\Theta(\tau_n)$ intersects with all lines $\Theta(\tau_m)$ for $m<n$, at $c_{n,m}$. Because of monotonicity, the smallest of these points of intersection is $c_{n,n-1}$. Thus $\tau_n$ is the optimal path for $c<c_{n,n-1}$. $\Theta(\tau_{n-1})$ intersects with all lines $\Theta(\tau_m)$ for $m<n$ at $c_{n-1,m}$. The first point of intersection is $c_{n-1,n-2}$. So $\tau_{n-1}$ is replaced by $\tau_{n-2}$, and in turn by $\tau_{n-3}$ and so on.
\end{proof}

\section{All optimal networks are paths for sufficiently small $\gamma$} \label{sec:maintheorem}
Now we prove that for $c_{m+1,m}<c<c_{m,m-1}$, networks with a unit source-sink
pair optimizing $H+cD$  are approximately paths of length $m$ in the limit as $\gamma\to 0$. Our proof works for any subinterval of $\left(c_{m+1,m},c_{m,m-1}\right)$. The parameter $0<\sigma<1$ represents the fraction of $[c_{m+1,m},c_{m,m-1}]$ covered by the subinterval. We can also represent the subinterval by $c_{m+1,m}+\rho<c<c_{m-1,m}-\rho$, where $\rho = \frac{1}{2}(1-\sigma) (c_{m-1,m}-c_{m+1,m})$.
\begin{thm}
\label{thm:main-theorem}
Let $G$ be an ambient network with a single unit-flow source and sink.
Let $m$ be a possible length of a path in $G$ connecting the source
to the sink. Let $0<\sigma<1$. We claim that there exists $\epsilon>0$ and $\Gamma>0$ such that if 
 $\delta>0$ and $\mathcal{F_{\delta}}$ is the network of strong nodes such that
$\left|H_{\mathcal{F}_{\delta}}-H\right|<\epsilon$ provided by by Theorem \ref{thm:strong-node-theorem}, and the material cost exponent $\gamma<\Gamma$, then for any $c_{m+1,m}+\rho<c<c_{m,m-1}-\rho$
the network $\mathcal{F}_\delta$ is a path of length $m$.
\end{thm}
\begin{proof}
We can simplify the calculations in our proof by appealing to the result from Section \ref{sec:Cinvariance}, that the sequence of optimizers is identical as $c$ is varied for any value of $C$. Accordingly we consider the special case $C = m-1$. For this choice of material cost $C$, $D(\tau_n) = (m-1)^{-1/\gamma}(n-1)^{1+1/\gamma}$. Then the computation of $\Theta$ on paths is is drastically simplified:
\begin{equation}
D(\tau_n) \to \left\{\begin{array}{cc} 0 & \hbox{if~} n<m \\ m-1 & \hbox{if~} n=m \\ \infty & \hbox{if~} n>m~. \end{array} \right.
\end{equation}
The $\Theta$ loci of $\tau_{m-1}$, $\tau_m$, $\tau_{m+1}$ are shown in Fig. \ref{fig:thetaloci}. Define $\rho$ as above, for the fixed material cost $C=m-1$.
Let $\Gamma$ be such that we can choose $\epsilon$ with $\Theta(\tau_m)+\epsilon<\Theta(\tau_{m\pm 1})$ for all $c_{m+1,m}+\rho<c<c_{m,m-1}-\rho$ and $\gamma<\Gamma$. Let $q_{ij}$ be a flow on $G$ with the specified source and sink. According to Theorem \ref{thm:strong-node-theorem}, we can define $\delta>0$ a dominance factor and $\mathcal{F}_\delta$ a network of strong nodes such that $\left|H(\mathcal{F}_\delta) -H\left(q_{ij}\right)\right|<\epsilon$. $\mathcal{F}_\delta$ has no leaf nodes except, potentially, the source and the sink.
\begin{figure}
\begin{center}
\includegraphics[width=0.8\textwidth]{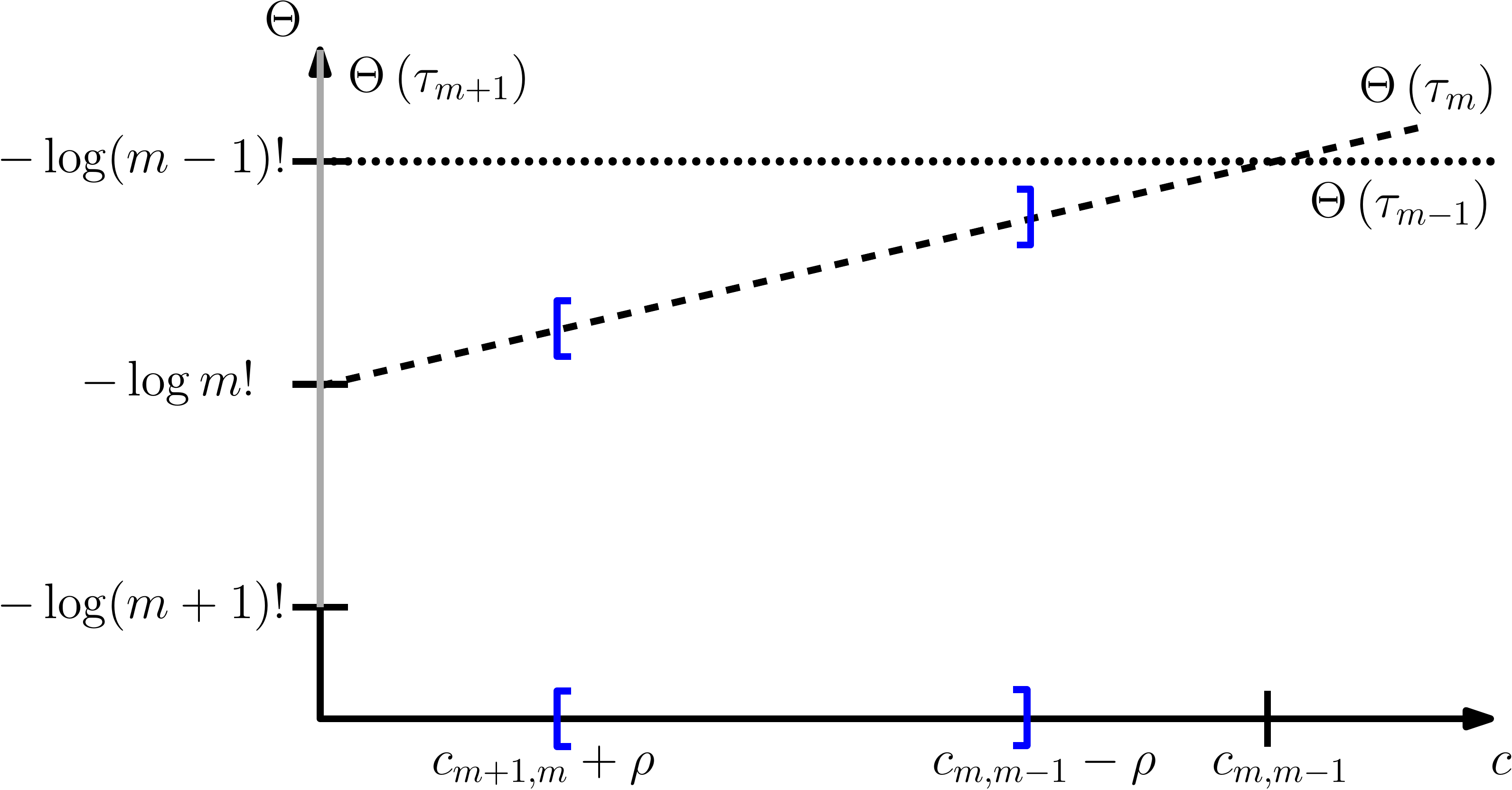}
\caption{$\Theta$-against-$c$ loci for the simple paths $\tau_{m-1}$ (dotted, black), $\tau_m$ (dashed, black), $\tau_m$ (solid, gray), in the limit as $\gamma\to 0$.} \label{fig:thetaloci}
\end{center}
\end{figure}
Suppose $\mathcal{F}_\delta$ contains $n$ nodes. Then by Euler's topological formula, it must contain $n-1+F$ edges, where $F$ is the number of faces in the network (given the constraints on $\mathcal{F}_\delta$, $F=0$ if and only if $\mathcal{F}_\delta$ is a path). Hence $-H(\mathcal{F}_\delta)\geq \log n!$. Each edge must carry, at minimum, flow $\delta^{n-1+F}$. Accordingly, the dissipation in the network can be bounded below by $D^*(n) = \delta^{2(n-1+F)} \left(\frac{n-1+F}{m-1}\right)^{1/\gamma}$ by Theorem \ref{thm:Murray}. $D^*(n)\to\infty$ as $\gamma\to 0$ if $n+F>m$. Since networks with bounded dissipation exist, the optimal network must have $n+F\leq m$. We can then compare the strong network with $\tau_{m-1}$, and $\tau_{m}$. If $n\leq m-1$, by Corollary \ref{prop:Best-for-restricted}, $\Theta(\mathcal{F}_\delta)\geq -H( \mathcal{F}_\delta)\geq -\log n! \geq -\log(m-1)! = \lim_{\gamma\to 0}\Theta (\tau_{m-1})$. This implies that $\Theta(q_{ij})>\Theta(\tau_{m})$, and so $q_{ij}$ is in fact sub-optimal. Therefore $n=m$. Since $n+F\leq m$, $F=0$, i.e. $\mathcal{F}_\delta$ is a path of length $m$.
\end{proof}
By choice of $\epsilon$, $\mathcal{F}_\delta$ can approximate $\mathcal{F}$ arbitrarily closely in $\Theta$. Since our convergence result can be made uniform in $\epsilon$ over the interval $c_{m+1,m}+\rho<c<c_{m,m-1}-\rho$, it follows that $\mathcal{F}$ converges to some path $\tau_m$, as $\gamma\to 0$, except possibly at the points $c_{m,m-1}$. Our proof method does not provide us with a way to prove convergence at these points, but based on our numerical simulations, we think it is likely that as $\gamma\to 0$ there are simply two optima, $\tau_{m}$ and $\tau_{m-1}$, with indistinguishable $\Theta$ values at these crossover $c$-values. 

\section{Discussion}
\label{sec:conclusions}

We introduced and analyzed theoretically and by numerical simulations two measures of mixing quality on networks, one measuring the diversity of places within a network that may be reached by cues originating within that network (sender entropy), and the other reflecting the diversity of cues that are received at each point within the network (receiver entropy). Happily, we were able to show that sender entropy for a network is equivalent to the receiver entropy on the same network if flows are reversed, allowing us to focus on optimizing just one kind of entropy within this paper. The mixing entropy quantifies the diversity of signals, which may include cues, genotypes and nutrients present at each point within the network: it is important to determine which type of mixing an network may be prioritizing before comparing it to theoretical calculations. Importantly, while at small $\gamma$ optimizing either entropy will produce identical networks, at biologically relevant values of $\gamma$, which type of mixing is most important to the network influences important features of its organization, such as the placement of loops. 

We proved that in the single source-single sink geometry, the optimal networks converge to simple paths joining source to sink, with the path length determined by the different priorities that the network gives to mixing (which favors long paths) and to dissipation (which favors short paths). Intriguingly, our numerical simulations suggest that there is a finite value of $\gamma$, which for our $5\times5$ ambient grids is approximately 0.45, at which optimal networks transition from loopy structures to simple path (see Fig. \ref{fig:gammavaries}). However, it is hard to guarantee that the network does not contain weak edges that do not show up in Fig. \ref{fig:triangle_1src1snk_5_q}. Our optimization algorithm enforces positivity of conductance on all edges, only to filter low conductance edges at the end. Accordingly it is not readily able to distinguish between small, but finite conductances that vanish only as $\gamma\to 0$, and a bifurcation that removes edges at a finite value of $\gamma$, so we must be cautious about interpreting the disappearance of loops at finite $\gamma$ as evidence of a phase transition in the network, analogous to the disappearance of loops at $\gamma=1$ for dissipation-minimizing networks \cite{bohn2007structure}. 

We have narrowly focused on the case where there is only one source and one sink within the network, allowing us to rigorously validate our numerical results. However, our numerical optimization method is equally applicable to networks with multiple sources and sinks, and it is possible to explore for example the conditions under which a network that is transporting material from a pair of sources to a pair of sinks, will determine to maintain two separate flows, or bring these flows together \cite{PhDThesis}. In particular, determining whether real network forming organisms such as fungi and slime molds have mixing-optimizing networks will require that we properly model the locations of the sources and sinks that drive their flows.

It is equally important when comparing optimal mixing networks with real biological networks to pin down the value of $\gamma$ for these networks. On theoretical grounds, we expect real biological networks in which vessels are simple tubes to operate in a range of $\gamma$ from $\gamma=1/4$ (when the cost of network upkeep is proportional to the surface area of its vessels) to $\gamma=1/2$ (when upkeep is proportional to vessel volume). Although direct measurement of $\gamma$ is impossible, the branching hierarchies of xylem vessels in plants and some levels of cellular tubes in the slime mold \textit{Physarum polycephalum} obey Murray's law (Theorem \ref{thm:Murray}) \cite{PlantTransportMcCullohAdler,akita2016experimental}. Specifically if $Q\propto \kappa^{\frac{\gamma+1}{2}}$, then $\sum \kappa^{\frac{\gamma+1}{2}}$ will be conserved between different levels of a hierarchical network. For simple tubes, we may assume the Hagen-Poiseuille law (that the conductance of a vessel and its radius, $a$ are related by $\kappa\propto a^4$), it follows that $\sum a^{2\gamma+2}$ is conserved. 

In \textit{P. polycephalum}, $\sum r^\alpha$ is conserved across different levels of the hierarchy, with a range of $\alpha$ values between $2.5-3.3$ \cite{akita2016experimental}, corresponding to $0.27<\alpha<0.65$. $\alpha$ values are similar for plants, but determining $\gamma$ from vessel radii is complicated by the fact that the xylem vessels (like the cords of mycorrhizal fungal networks) are constituted of many smaller tubes. Suppose these tubes have radius $A$ but individually obey the Hagen-Poiseuille law, then: $\kappa\propto a^2A^2$. The total number of tubes at the same level in the hierarchy is reported to increase by a factor $F\approx 1.2$ moving from larger to smaller tubes \cite{PlantTransportMcCullohAdler}. Accordingly, since $a$ decreases by a factor of $2^{-1/\alpha}$ when one tube splits into two then $A\propto a^{1-\frac{\alpha}{2}+\frac{\alpha}{2}\log_2F}$, and so $\alpha = (\gamma+1)\left(2-\frac{\alpha}{2}+\frac{\alpha}{2}\log_2F\right)$. We therefore estimate that the plants in \cite{PlantTransportMcCullohAdler} have $\gamma$ values ranging from  0.8 for \textit{Fraxinus pensylvanica} ($F=1.2$, $\alpha=2.2$) up to 1.4 for \textit{Campsis radicans} ($F=1.4$, $\alpha=3$). So slime molds span the value of $\gamma\approx 0.5$ at which our calculations show loops being eliminated from the optimal network, while plant networks sit high above this value.

The optimal networks shown in Figure \ref{fig:gammavaries} for $\gamma\gtrsim 0.5$, qualitatively resemble the real structures of migrating slime mold networks, in which densely interconnected `fans' of tubes are linked together by sparsely connected or even loopless networks (see e.g. Fig. 1 in \cite{alim2018fluid}).  In future work, we plan to analyze the optimal loopy networks found by our algorithm to determine why optimal mixing requires fans (loopy regions) don't appear throughout the network but are located only near the source, as well as to understand how the tradeoffs between mixing and dissipation can be used to predict the size of the fan relative to the total length of the network.

That real network forming organisms do not form tours may result from their $\gamma$ values being too high. However, even at low values of $\gamma$ networks face other tradeoffs, such as resistance to damage and or the need to minimize dissipation when the sources and sinks fluctuate in strength \cite{katifori2010damage}. An additional property that must be highlighted for organisms such as fungi and slime molds that have indeterminate growth is that organisms need to maintain their mixing while growth pushes sources and sinks ever further apart. A tour can be extended indefinitely to include to an ever increasing number of nodes by extending it node by node. However, this model of growth extends the network only by adding a single edge at a time, restricting growth to a single growing tip and is an inefficient strategy for a fungus or other foraging organism, that must compete for space and resources with other organisms. The type of network formed by a network is also shaped by the constraints on how it must form this network. Fast foraging may favor growth in multiple directions simultaneously, facilitated by the organism having multiple growing tips. Thus optimization principles such as those developed in this paper only achieve true biological relevance when linked to a set of rules that a growing organism can follow to attain the optima. Such rules have been only recently elucidated for dissipation minimizing networks (see e.g. \cite{hu2013adaptation}), leaving unmet the challenge of constructing rules to achieve more complex objectives, including mixing.

\section*{Acknowledgments}
We thank Karen Alim, Eleni Katifori and Sebastien Roch for many useful discussions at a sequence of Square Meetings hosted by the American Institute for Mathematics, where the idea for this project was developed.  

\appendix
\section{Computation of derivatives of $\Theta$} \label{sec:appendix_derivs}

To differentiate $\Theta$ we compute the all of intermediate variables appearing in Eq.(\ref{eq:composition}): i.e. $p_{i},q_{ij},f_{i},T_{ij},P_{ij},\tilde{q}_{ij}$ and
$N_{i}$. The pressures $p_{i}$ are first obtained by solving Eqn. \ref{eqn:Poiss_eqn}, using the Matlab function mldivide. We then solve a chain of equations to obtain the Lagrange multipliers: 
\begin{eqnarray}
 \xrightarrow{\frac{\partial\Theta}{\partial N_i}=0}  \alpha  
\xrightarrow{\frac{\partial\Theta}{\partial\tilde{q}_{ab}}=0} \gamma
\xrightarrow{\frac{\partial\Theta}{\partial P_{ab}}=0} \mu  \xrightarrow{\frac{\partial\Theta}{\partial T_{ab}}=0} \lambda
\xrightarrow{\frac{\partial\Theta}{\partial f_{a}}=0} \beta \xrightarrow{\frac{\partial\Theta}{\partial p_{a}}=0} \nu~.
\end{eqnarray}
First:
\begin{equation}
\frac{\partial\Theta}{\partial N_{a}}  =  \sum_{i}\left(-\frac{\tilde{q}_{ia}}{N_{a}^{2}}\log\left(\frac{\tilde{q}_{ia}}{N_{a}}\right)-\frac{\tilde{q}_{ia}}{N_{a}^{2}}\right)-\alpha_{a}=0.
\end{equation}
Second:
\begin{equation}
\frac{\partial\Theta}{\partial\tilde{q}_{ab}}  =  \frac{f_{b}}{N_{b}}\log\left(\frac{\tilde{q}_{ab}}{N_{b}}\right)+\frac{f_{b}}{N_{b}}+\alpha_{b}-\gamma_{ab}=0
\end{equation}
Third:
\begin{equation}
\frac{\partial\Theta}{\partial P_{ab}}  =  \mu_{ab}-\sum_{i}\mu_{ib}T_{ia}+\gamma_{ab}f_{a}=0
\end{equation}
so:
\begin{equation}
\gamma_{ab}  =  \left(\sum_{i}\mu_{ib}T_{ia}-\mu_{ab}\right)/f_{a}
\end{equation}
Fourth:
\begin{equation}
\frac{\partial\Theta}{\partial T_{ab}}  =  -\sum_{j\in\mathcal{N}}\mu_{aj}P_{bj}-\lambda_{ab}=0
\end{equation}
Fifth:
\begin{equation}
\frac{\partial\Theta}{\partial f_{a}}  =  \sum_{j:\tilde{q}_{ja}>0}\frac{\tilde{q}_{ja}}{N_{a}}\log\left(\frac{\tilde{q}_{ja}}{N_{a}}\right)+\sum_{j\in\mathcal{N}}\gamma_{aj}P_{aj}-\sum_{j\in n(i)}\lambda_{aj}\frac{q_{aj}\mathbf{1}_{q_{aj}>0}}{f_{a}^{2}}-\beta_{a}=0~.
\end{equation}
Sixth, to calculate $\frac{\partial \Theta}{\partial p_a}$ we make use of the results $\frac{\partial}{\partial p_{a}}q_{ai}=\kappa_{ai}$
and $\frac{\partial}{\partial p_{a}}q_{ia}=-\kappa_{ia}$. Thus: 
\begin{eqnarray}
\frac{\partial\Theta}{\partial p_{a}} & = & \sum_{j}\kappa_{aj}(\nu_{a}-\nu_{j})+\sum_{i\in n(a)}\left(\beta_{a}\kappa_{ai}\mathbf{1}_{q_{ai}>0}-\beta_{i}\kappa_{ai}\mathbf{1}_{q_{ia}>0}\right)\nonumber \\ & & +\sum_{i\in n(a)}\left(\lambda_{ai}\frac{\kappa_{ai}\mathbf{1}_{q_{ai}>0}}{f_{a}}-\lambda_{ia}\frac{\kappa_{ai}\mathbf{1}_{q_{ia}>0}}{f_{i}}\right) = 0.
\end{eqnarray}
Thus solving for the Lagrange multipliers $\nu_i$ requires
solving a Poisson equation on the network similar to Eqn. \ref{eqn:Poiss_eqn}.

\section{Finding adjacent flow topologies} \label{sec:appendix_ShermanMorrison}

We assume that the network of non-zero conductances
has a single connected component, because although very small conductances are treated as negligible throughout our algorithm, they remain large enough to keep the Laplacian rank complete. We take the inverse of the version of the Laplacian defined in Section \ref{sec:movement}, $\tilde{\Delta}_\kappa$ for the initial network. We compute the directions of flow on each edge within the network (edges with low flows are ignored). The set of networks with the same directions of flow constitutes one of the watersheds shown in Fig \ref{fig:mixinglandscape}. We systematically vary one conductance $\kappa_{ab}$ within the network to find an adjacent watershed -- i.e. a flow network in which some subset of the non-negligible flows have been reversed. We find the threshold values for $\kappa_{ab}$ at which one or more flow directions are reversed, by appealing to the Sherman-Morrison formula \cite{sherman1950adjustment} (we thank Eleni Katifori for bringing the S.M. formula to our attention). Specifically, if the conductance in edge $(a,b)$ is increased to $\kappa_{ab}+t$, then the Laplacian for the new network becomes
\begin{equation}
\tilde{\Delta}_{\tilde{\kappa}_{ij}}  =  \tilde{\Delta}_{\kappa_{ij}}+t(e_{a}-e_{b})(e_{a}-e_{b})^{T}.
\end{equation}
 Then the Sherman-Morrison formula yields
\begin{eqnarray}
\tilde{\Delta}_{\tilde{\kappa}_{ij}}^{-1} & = & \left(\tilde{\Delta}_{\kappa_{ij}}+t(e_{a}-e_{b})(e_{a}-e_{b})^{T}\right)^{-1}\nonumber\\
 & = & \tilde{\Delta}_{\kappa_{ij}}^{-1}-\frac{\tilde{\Delta}_{\kappa_{ij}}^{-1}t(e_{a}-e_{b})(e_{a}-e_{b})^{T}\tilde{\Delta}_{\kappa_{ij}}^{-1}}{1+(e_{a}-e_{b})^{T}\tilde{\Delta}_{\kappa_{ij}}^{-1}t(e_{a}-e_{b})}.
\end{eqnarray}

Given another edge $(u,v)$, We wish to find a perturbation
to $\kappa_{ab}$ such that the flow along $(u,v)$ is reversed. Let $R_{i}$ be the $i^{\text{th}}$ row of $\tilde{\Delta}_{\kappa_{ij}}^{-1}$,
and $d_{ij}$ be the $i,j$ entry of $\tilde{\Delta}_{\kappa_{ij}}$.
Then the pressure drop is given by:
\begin{equation}
\tilde{p}_{u}-\tilde{p}_{v}  =  \left(R_{u}-R_{v}\right)Q-\frac{t\left(d_{au}-d_{av}-d_{bu}+d_{bv}\right)(R_{a}-R_{b})Q}{1+t\left(d_{aa}-d_{ab}-d_{ba}+d_{bb}\right)}.
\end{equation}
Therefore the pressure drop is a monotonic function of $t$ so the
the zero of this equation is where the pressure reverses. Setting
the left side to $0$ we get the value $t=t_{abuv}$ at which flow reversal occurs:
\begin{equation}
t_{abuv} \frac{p_{u}-p_{v}}{\left(d_{au}-d_{av}-d_{bu}+d_{bv}\right)(p_{a}-p_{b})-\left(d_{aa}-d_{ab}-d_{ba}+d_{bb}\right)(p_{u}-p_{v})}~ \label{eq:flowreversal}.
\end{equation}

\bibliographystyle{siamplain}
\bibliography{ref}
\end{document}